\documentclass[runningheads]{llncs}
\usepackage[sectionbib,numbers,sort&compress]{natbib}
\usepackage[T1]{fontenc}
\usepackage{graphicx}
\usepackage[colorlinks]{hyperref}
\usepackage{color}

\usepackage{amssymb}
\usepackage{amsmath}
\usepackage{bm}
\usepackage{bbm}
\usepackage{stmaryrd}
\usepackage{booktabs}
\usepackage{multirow}
\usepackage{threeparttable}
\usepackage{textcomp}
\usepackage[small]{complexity}
\usepackage{adjustbox}
\usepackage{listings}
\usepackage[chapter]{algorithm}
\usepackage{algorithmicx}
\usepackage{algpseudocode}
\usepackage{physics}
\usepackage[inline]{enumitem}
\usepackage{tabularx}
\usepackage{makecell}
\usepackage{xparse}
\usepackage{mathtools}
\usepackage{relsize}
\usepackage{todonotes}
\usepackage{stackengine}
\usepackage{bookmark}
\usepackage{changepage}
\usepackage{placeins}
\usepackage{siunitx}
\usepackage{thmtools,thm-restate}

\usepackage{crossreftools}
\usepackage{tikz}
\usepackage{pgfplots}
\pgfplotsset{compat=1.15}
\usetikzlibrary{shapes.misc}
\usetikzlibrary{arrows.meta}
\usetikzlibrary{arrows}
\usetikzlibrary{calc}
\usetikzlibrary{decorations.pathreplacing}
\usetikzlibrary{quotes}
\tikzset{cross/.style={cross out, draw=black, minimum size=2*(#1-\pgflinewidth),
    inner sep=0pt, outer sep=0pt},
cross/.default={1pt}
}

\usepackage[capitalize]{cleveref}   

\bookmarksetup{depth=3}
\crefname{equation}{}{}
\Crefname{equation}{Eq.}{Eqs.}
\ExplSyntaxOn
\cs_set_eq:Nc \bers_cref:nn { @cref }
\cs_generate_variant:Nn \bers_cref:nn { nx }
\cs_set_protected:cpn { @cref } #1 #2
 {
  \seq_set_split:Nnn \l_bers_cref_seq { , } { #2 }
  \bers_cref:nx { #1 } { \seq_use:Nn \l_bers_cref_seq { , } }
 }
\seq_new:N \l_bers_cref_seq
\ExplSyntaxOff

\ExplSyntaxOn   
\NewDocumentCommand{\mathlist}{m}
{
    \clist_set:Nn \l_tmpa_clist { #1 }
    \tl_set:Nn \l_tmpa_tl { \clist_use:Nn \l_tmpa_clist { ,\, } }
    \ensuremath { \l_tmpa_tl }
}
\ExplSyntaxOff
\newcommand{\defeq}{\coloneq}
\newcommand{\condSet}[2]{\left\{#1 \;\middle|\; #2\right\}}
\DeclareMathOperator*{\argmax}{argmax^*}    
\DeclareMathOperator*{\argmin}{argmin^*}
\DeclareMathOperator*{\argmaxS}{argmax}     
\DeclareMathOperator*{\argminS}{argmin}
\DeclareMathOperator{\sign}{sign}

\DeclareMathOperator{\diag}{diag}
\DeclareMathOperator{\fGr}{Gr}  
\DeclareMathOperator{\dom}{dom}
\newcommand{\signf}[1]{\sign\sqty{#1}}
\newcommand{\simplex}[1]{\Delta_{#1}}
\newcommand{\maxq}[1]{\max\qty{\mathlist{#1}}}  
\newcommand{\minq}[1]{\min\qty{\mathlist{#1}}}  
\newcommand{\sqty}[1]{\qty(\mathlist{#1})}  
\newcommand{\cqty}[1]{\qty{\mathlist{#1}}}  
\newcommand{\lrmid}[2]{\left(\mathlist{#1} \;\middle|\; \mathlist{#2}\right)}
\newcommand{\card}[1]{\left|#1\right|}  
\newcommand{\indicator}[1]{\mathbbm{1}_{#1}}  
\newcommand{\numSym}[2]{
    \stackanchor[1pt]%
    {\scalebox{.5}{\raisebox{.5pt}{\textcircled{\raisebox{-.9pt}{#2}}}}}%
    {\scalebox{.7}{\ensuremath{#1}}}
}
\newcommand{\eqN}[1]{\mathrel{\numSym{=}{#1}}}
\newcommand{\leqN}[1]{\mathrel{\numSym{\leq}{#1}}}
\newcommand{\geqN}[1]{\mathrel{\numSym{\geq}{#1}}}

\newcommand{\floor}[1]{\left\lfloor#1\right\rfloor}

\DeclareMathOperator{\setCl}{cl}
\DeclareMathOperator{\setConv}{conv}

\renewcommand{\real}{\mathbb{R}}

\newcommand{\posReal}{\mathbb{R}_+}

\newcommand{\integer}{\mathbb{Z}}
\newcommand{\posInt}{\integer_+}
\newcommand{\extPosInt}{\overline{\integer}_+}

\newcommand{\setN}{\mathcal{N}}  
\newcommand{\setS}{\mathcal{S}}  
\newcommand{\setR}{\mathcal{R}}
\newcommand{\setM}{\mathcal{M}}

\newcommand{\Ry}{Y}         
\newcommand{\vRy}{\V{\Ry}}
\newcommand{\Ryf}[2][]{\Ry#1\sqty{#2}}
\newcommand{\vRyf}[2][]{\vRy#1\sqty{#2}}

\newcommand{\gRy}{g_{\vRy}}  
\newcommand{\gRyf}[1]{g_{\vRy}\sqty{#1}}

\newcommand{\oslpGame}{OLP-Game}
\newcommand{\aoslpGame}{an \oslpGame{}}
\newcommand{\tfzGame}{\textsc{TZolp}-Game}
\newcommand{\ppfuncSym}[1][]{\alpha_{\operatorname{#1}}}
    \newcommand{\ppfunc}[3][]{\ppfuncSym[#1]\qty(#2,\,#3)}
\newcommand{\ppfMaxC}[1]{\bar{c}\qty(#1)}  
\newcommand{\ppfuncInv}[3][]{\ppfuncSym[#1]^{-1}\qty(#2,\,#3)}
\newcommand{\maxPayoffSym}{\mathcal{P}^+}
\newcommand{\minPayoffSym}{\mathcal{P}^-}
\DeclareMathOperator{\ncSetSym}{N_C}    
\newcommand{\ncSet}[1]{\ncSetSym\sqty{#1}}
\NewDocumentCommand{\lpgPlayerParam}{ o m }{
    \IfNoValueTF{#1}{
        \def\lpgPlayerParamStg{\vsig_{#2}}
    } {
        \def\lpgPlayerParamStg{#1}
    }
    \lrmid{v, c, \lpgPlayerParamStg}{\vsig_{-#2}}
}

\newcommand{\lpgCharSetSym}{\mathcal{C}}
\newcommand{\lpgCharSetP}[2]{       
    \lpgCharSetSym\lrmid{#1}{#2}}

\newcommand{\maxPayoff}[2][]{\maxPayoffSym_{#1}\sqty{#2}}
\newcommand{\minPayoff}[2][]{\minPayoffSym_{#1}\sqty{#2}}
\newcommand{\uncPayoffSym}{\mathcal{U}}   
\newcommand{\uncPayoff}[1]{\uncPayoffSym\sqty{#1}}

\newcommand{\lowrank}[2]{\widetilde{#1}^{(#2)}}   
\newcommand{\lowrankT}[2]{\widetilde{#1}^{(#2)\,\intercal}} 
\newcommand{\lowrankSet}[2]{\operatorname{LR}\sqty{#1, #2}}   

\newcommand{\intInt}[2]{\left\llbracket #1,\,#2 \right\rrbracket}

\newcommand{\lineRef}[1]{\hyperref[#1]{Line~\ref*{#1}}}

\newcommand{\funSubdGSym}[2][\epsilon]{\partial_{#1}#2}
\newcommand{\funSubdDSym}[2][\epsilon]{\mathring{D}_{#1}#2}

\newcommand{\funSubdP}[2][\epsilon]%
{\qty(\funSubdDSym[#1]{#2},\,\funSubdGSym[#1]{#2})}

\newcommand{\V}[1]{\bm{#1}}
\newcommand{\T}[2][]{#2^{#1\intercal}}
\newcommand{\xTx}[1]{\T{#1}#1}

\newcommand{\Null}[1]{\mathcal{N}\sqty{#1}}
\newcommand{\leftNull}[1]{\mathcal{N}\sqty{\T{#1}}}
\newcommand{\NullT}[1]{\T{\mathcal{N}}\sqty{#1}}
\newcommand{\leftNullT}[1]{\T{\mathcal{N}}\sqty{\T{#1}}}
\newcommand{\mcol}[2]{#1_{:,\,#2}}  
\newcommand{\spn}[1]{\operatorname{span}\sqty{#1}}

\renewcommand{\va}{\V{a}}
\newcommand{\vA}{\V{A}}

\renewcommand{\vb}{\V{b}}
\newcommand{\vB}{\V{B}}

\newcommand{\vd}{\V{d}}

\newcommand{\ve}{\V{e}}

\newcommand{\vI}{\V{I}}

\newcommand{\vM}{\V{M}}

\newcommand{\vP}{\V{P}}

\newcommand{\vQ}{\V{Q}}
\newcommand{\vq}{\V{q}}

\newcommand{\vR}{\V{R}}
\newcommand{\vS}{\V{S}}
\newcommand{\vU}{\V{U}}
\renewcommand{\vu}{\V{u}}
\newcommand{\vV}{\V{V}}
\newcommand{\vv}{\V{v}}

\newcommand{\vx}{\V{x}}
\newcommand{\vy}{\V{y}}

\newcommand{\vX}{\V{X}}
\newcommand{\vY}{\V{Y}}
\newcommand{\vSig}{\V{\Sigma}}
\newcommand{\vsig}{\V{\sigma}}

\newcommand{\vs}{\V{s}}

\newenvironment{enuminline}
{\begin{enumerate*}[label=(\roman*),itemjoin={{; }},itemjoin*={{; and }}]}
{\end{enumerate*}}

\algnewcommand\algorithmicinput{\textbf{Input:}}
\algnewcommand\Input{\item[\algorithmicinput]}
\makeatletter
\NewDocumentCommand{\LeftComment}{s m}{%
  \Statex \hspace*{\ALG@thistlm}\(\triangleright\) #2}
\makeatother

\makeatletter
\newcommand\notsotiny{\@setfontsize\notsotiny\@vipt\@viipt}
\makeatother

\newcommand{\overbar}[1]{%
\mkern 1.5mu\overline{\mkern-1.5mu#1\mkern-1.5mu}\mkern 1.5mu}

\newenvironment{separateProof}[1]{   
\noindent {\itshape Proof of \cref{#1}.}
}{
}

\spnewtheorem*{proofidea}{Proof idea}{\itshape}{\rmfamily}
\begin{document}
\title{Limited-perception games}
%
%
\author{Kai Jia\orcidID{0000-0001-8215-9899} \and
Martin Rinard
}
%
%
\institute{Department of Electrical Engineering and Computer Science, \\
Massachusetts Institute of Technology, USA \\
\email{\{jiakai,rinard\}@csail.mit.edu}}
\maketitle              
\begin{abstract}
    We study rational agents with different perception capabilities in strategic
    games. We focus on a class of one-shot limited-perception games. These games
    extend simultaneous-move normal-form games by presenting each player with an
    individualized perception of all players' payoff functions. The accuracy of
    a player's perception is determined by the player's capability level.
    Capability levels are countable and totally ordered, with a higher level
    corresponding to a more accurate perception. We study the rational behavior
    of players in these games and formalize relevant equilibria conditions. In
    contrast to equilibria in conventional bimatrix games, which can be
    represented by a pair of mixed strategies, in our limited perception games a
    higher-order response function captures how the lower-capability player uses
    their (less accurate) perception of the payoff function to reason about the
    (more accurate) possible perceptions of the higher-capability opponent. This
    response function characterizes, for each possible perception of the
    higher-capability player (from the perspective of the lower-capability
    player), the best response of the higher capability player for that
    perception. Since the domain of the response function can be exponentially
    large or even infinite, finding one equilibrium may be computationally
    intractable or even undecidable. Nevertheless, we show that for any
    $\epsilon$, there exists an $\epsilon$-equilibrium with a compact, tractable
    representation whose size is independent of the size of the response
    function's domain. We further identify classes of zero-sum
    limited-perception games in which finding an equilibrium becomes a
    (typically tractable) nonsmooth convex optimization problem.

\end{abstract}

\section{Introduction}
\label{sec:introduction}

In many scenarios, players exhibit inherent limitations in various aspects of
their ability to generate maximally rational game play. Such \emph{capability
limitations} play an important role in previously considered scenarios
featuring, for example, limited strategy complexity \citep{ neyman1985bounded,
halpern2015algorithmic, yang2022on}, limited randomness in implementing mixed
strategies \citep{ neyman2000repeated, valizadeh2019playing,
orzech2023randomness}, limited accuracy in implementing strategies
\citep{selten1988reexamination, farina2018trembling}, and limited information
regarding players' payoffs~\citep{harsanyi1967games, aghassi2006robust,
kroer2018robust}.
\nopagebreak

We consider limited \emph{perception capabilities}, specifically limits in the
ability of players to accurately perceive (or predict) the payoffs of different
strategy profiles. Such limits affect the ability of the players to reason
effectively about the consequences of different strategies and therefore their
ability to maximize the rationality of their play. We present and study a new
model, \emph{one-shot limited-perception games} (\oslpGame{}s), in which each
player has a limited, hierarchically-ordered, capability to perceive the true
game payoffs:
\begin{enumerate}
    \item There is a normal-form \emph{true game}, hidden from the players, that
        determines the final payoffs of the players given a strategy profile.
    \item Each player has a \emph{capability level}; capability levels are
        countable and totally ordered, and are thus represented as integers.
    \item Each player perceives an \emph{abstraction} of the true game obtained
        by applying a \emph{payoff perception function} at this player's
        capability level to the true payoff functions of all the players. A
        higher capability level induces a more accurate perception. Each
        abstraction is still a normal-form game.
\end{enumerate}

The payoff perception function plays a central role in this model. We study
payoff perception functions characterized by three axioms (\cref{ def:lpg:ppf}):
\begin{enumerate}
    \item \emph{Capability path independence}: an abstraction obtained by
        applying the payoff perception function multiple times with different
        capability levels in an arbitrary order is the same as applying the
        perception payoff function once with the minimum capability level.
    \item \emph{Perfect perception with infinite capability}: any payoff
        function can be perfectly perceived by a player with an infinite
        capability level.
    \item \emph{Bounded concretization}: given any abstraction, the payoff
        values of compatible true games are bounded.
\end{enumerate}
The first axiom ensures that a single capability level (rather than a path of
capability levels) determines a perception, which induces a hierarchical
structure over the perception accuracy (e.g., a higher-capability player can
predict a lower-capability player's perception). The other two axioms induce
desirable properties of true payoff bounds (\cref{
thm:lpg:payoff-bounds-convexity}). This model captures a range of games with
limited payoff perceptions (see
\crefrange{ex:lpg:masked-payoff}{ex:lpg:limited-rank}).

We consider players with the common knowledge that every player knows
\begin{enuminline}
    \item the payoff perception function
    \item the capability levels of all players
    \item that each player maximizes their worst-case expected payoff given
        their limited perception of the true payoff function.
\end{enuminline}
The players use their (limited) perception of the true payoffs to reason about
the incentives and resulting actions of all players in the game. In the
two-player case with the first player having a lower capability, because the
first player has only limited knowledge of the second player's perception, and
must therefore consider all possible perceptions of the second player, a pair
$\sqty{\vx^*, \vy^*}$ of mixed strategies as in conventional bimatrix games
\citep{ nash1950equilibrium} cannot capture the equilibrium behavior in our
model; one particular issue is that the first player cannot solve $\vy^*$ using
their private information.

\nopagebreak
We instead propose a new equilibrium formalization $\sqty{\vx^*, \vRy}$ (\cref{
def:lpg:2p-fully-capability-aware-nash}), where $\vx^*$ is a mixed strategy of
the first (i.e., lower-capability) player, and $\vRyf{\cdot}$ is a higher-order
response function that maps each possible perception of the second player (from
the first player's perspective) to a mixed-strategy best response of the second
player for that perception. Then $\vx^*$ maximizes the first player's worst-case
payoff  against all possible responses of the second player.

It may not be immediately clear that such equilibria exist for all \oslpGame{}s.
We show that payoff bounds derived by reasoning about the possible true payoffs
given an abstraction produced via a payoff perception function are continuous
and concave in each player's strategy (\cref{ thm:lpg:payoff-bounds-convexity}).
The correspondence (a correspondence is a set-valued function) from other
players' response functions to the set of a player's best response functions is
therefore upper hemicontinuous and convex-valued. By considering an infinite
dimensional space containing the Cartesian product of all players' response
function spaces, one can apply the generalized Kakutani fixed-point
theorem~\cite{ fan1952fixed} to prove that an equilibrium exists.

Because the domain of $\vRyf{ \cdot}$ can be exponentially large or even
infinite, finding an equilibrium can be computationally intractable or even
undecidable. Despite this hardness, we present a positive result that for any
$\epsilon > 0$, there exists an $\epsilon$-equilibrium with a compact, tractable
representation whose size is independent of $\card{\dom \vRy}$ (\cref{
thm:lpg:fully-capability-aware-nash-repr}). One consequence is that even if
$\vRyf{\cdot}$ is not computable by any \PSPACE{} algorithm, a trusted oracle
can still communicate an efficient representation of an $\epsilon$-equilibrium
$\sqty{\vx', \vRy'}$ to Turing-complete players such that $\vRyf[']{\cdot}$ is
efficiently computable.

We next consider two-player zero-sum \oslpGame{}s. In sharp contrast to zero-sum
normal-form games, which are efficiently solvable due to the minimax
theorem~\citep[Section 1.4.2]{ nisan2007algorithmic}, finding an equilibrium in
a zero-sum \oslpGame{} is as computationally hard as in the general-sum case
since the latter can be reduced to the former (\cref{
thm:lpg:tfzg-gsum-reduction}). Nevertheless, we generalize the minimax condition
to identify efficiently solvable classes of two-player zero-sum \oslpGame{}s
(\cref{ sec:lpg:maximin-attainable-examples}), where finding an equilibrium
becomes a nonsmooth convex optimization problem (\cref{ thm:lpg:maximin-equil})
that is tractable in many cases~\citep{ bubeck2015convex}.

Researchers have studied various models of limited information games \citep{
harsanyi1967games, hartline2010bayesian, chen2021defending,
grosshans2013bayesian, farina2021model}. One classic model is Bayesian games
\citep{ harsanyi1967games} where each player's payoff is affected by a private
type drawn from a commonly known prior distribution. By comparison, in the
\oslpGame{} model, a player has no private knowledge of their own exact true
payoff function, and there is no distribution assumption. The distribution-free
approach has been considered by other authors in studying robust games where
players maximize worst-case payoffs~\citep{aghassi2006robust, kroer2018robust}.
To our knowledge, this paper is the first to model a hierarchy of perception
capabilities in a hidden true game and the first to study rational behavior in
this context.

\section{The limited-perception game model}
\label{sec:lpg:game-model}
This section formalizes the game model introduced in \cref{sec:introduction}. Of
note, we typically use bold symbols for vectors (lowercase) and matrices
(uppercase) (e.g., $\va$ is a vector, $a_i$ is the $i$-th element of $\va$, and
$\va_i$ is a vector indexed by $i$). We use $\argmax$/$\argmin$ to denote an
arbitrary maximizer/minimizer, and $\argmaxS$/$\argminS$ to denote the complete
set of maximizers/minimizers.

\begin{definition}[One-shot limited-perception game]
    \label{def:lpg:game}
    A finite \emph{one-shot limited-perception game} (\emph{\oslpGame} for
    short) with $n$ players is a tuple \\
    $G = \sqty{\setN, (\setS_i)_{ i\in\setN}, (u_i)_{i\in\setN}, \ppfuncSym}$
    where:
    \begin{itemize}
        \item $\setN = \intInt{1}{n}$ is the set of players ($\intInt{a}{b}$
            denotes integers from $a$ to $b$ inclusive).
        \item $\setS_i$ is a finite set of pure strategies available to $i$.
        \item $u_i : \setS_1 \times \cdots \times \setS_n \to \real$ is the
            \emph{true payoff function} of player $i$.
        \item $\ppfuncSym$ is the \emph{payoff perception function} of the game
            as defined in \cref{def:lpg:ppf}.
    \end{itemize}
    For a given game $G$, there can be multiple \emph{instances}. In each
    instance, each player has a capability level that determines their
    perception of the payoff functions. \Cref{ def:lpg:gameplay} specifies the
    gameplay.
\end{definition}

\begin{definition}[Payoff perception function]
    \label{def:lpg:ppf}
    In \aoslpGame{} (\cref{def:lpg:game}), the \emph{payoff perception function}
    is a function $\ppfuncSym: \qty(U \times \extPosInt) \mapsto U$ where $U
    \defeq \real^{\setS_1 \times \cdots \times \setS_n}$ is the set of payoff
    functions, and $\extPosInt \defeq \posInt \cup \cqty{\infty}$ is the
    extended set of positive integers. If $v = \ppfunc{u}{c}$, then $v$ is
    called an \emph{abstraction} of $u$ at \emph{capability level} $c$. The
    payoff perception function satisfies the following properties:
    \begin{itemize}
        \item \emph{Capability path independence}: for any $u \in U$ and $c_1,
            c_2 \in \posInt$, it holds that \\
            $\ppfunc{\ppfunc{u}{c_1}}{c_2} = \ppfunc{u}{\minq{c_1, c_2}}$.
        \item \emph{Perfect perception with infinite capability}: for any $u \in
            U$, $\ppfunc{u}{\infty} = u$.
        \item \emph{Bounded concretization}: given any $v \in U$, there exists a
            constant $b \in \posReal$ such that for any $u \in U$ and $c \in
            \extPosInt$ satisfying $\ppfunc{u}{c} = v$, it holds that $\max_{s
            \in \setS_1 \times \cdots \times \setS_n} \abs{u(s)} \leq b$.
    \end{itemize}
    Let $\ppfMaxC{u} \defeq \min\condSet{c \in \extPosInt}{\ppfunc{u}{c} = u}$.
    The second property guarantees that $\ppfMaxC{u}$ is well-defined. The value
    $\ppfMaxC{u}$ is called the \emph{intrinsic capability} of the payoff
    function $u(\cdot)$ given the payoff perception function~$\ppfunc{
    \cdot}{\cdot}$.
\end{definition}

\begin{example}[Masked payoff game]
    \label{ex:lpg:masked-payoff}
    Consider a function $\ppfuncSym[s]: \real^{m \times n} \times \extPosInt
    \mapsto \real^{m \times n}$ where $\ppfunc[s]{\vM}{c}$ keeps the $c$ entries
    with the largest absolute values in $\vM$ and sets the rest to zero.
    Tie-breaking can be arbitrary but must be consistent with capability path
    independence (e.g., by preferring to keep the smallest row/column indices).
    Then $\ppfuncSym[s](\cdot)$ is a payoff perception function for bimatrix
    games.
\end{example}

\begin{example}[Quantized payoff game]
    \label{ex:lpg:quant-payoff}
    Consider a function that keeps only $c$ decimal digits for each payoff value
    by rounding towards zero. Formally, let $\ppfunc[q]{u}{c} = u'$, then for
    any strategy profile $\vs \in \setS_1 \times \cdots \times \setS_n$,
    \begin{align*}
        u'(\vs) \defeq \begin{dcases}
            \signf{u(\vs)} \cdot 10^{-c} \cdot \floor{10^c \cdot \abs{u(\vs)}}
            & \text{if } c < \infty \\
            u(\vs) & \text{if } c = \infty
        \end{dcases}
    \end{align*}
    It is easy to verify that $\ppfunc[q]{\cdot}{\cdot}$ is a payoff perception
    function.
\end{example}

\begin{example}[Limited-rank game]
    \label{ex:lpg:limited-rank}
    A limited-rank game is a two-player \oslpGame{}. Let $\sqty{\vA, \vB}$ be
    the true payoff matrices. A player at capability level $c$ perceives the
    game as $\sqty{\vA', \vB'}$ where $\vA'$ and $\vB'$ are the best rank-$c$
    approximations (minimizing the Frobenius norm of the residual matrix) of
    $\vA$ and $\vB$, respectively.
\end{example}

\begin{example}[Largest masked payoff game]
    If we change \cref{ex:lpg:masked-payoff} to keep the $c$ entries with the
    smallest absolute values, the resulting function is not a payoff perception
    function since it violates the bounded concretization property.
\end{example}

\begin{definition}[Gameplay of the \oslpGame]
    \label{def:lpg:gameplay}
    In \aoslpGame{} instance (\cref{def:lpg:game}), each player $i$ has a
    capability level $c_i \in \extPosInt$. The tuple $(c_i)_{i \in \setN}$ is
    called the \emph{capability profile} of this instance. Player $i$ perceives
    the payoff function of player $j$ as $v^i_j \defeq \ppfunc{u_j}{c_i}$. The
    game described by $(v^i_j)_{j\in\setN}$ is the \emph{perceived game} of
    player $i$. Players know the following as common knowledge: $\setN$,
    $(\setS_j)_{j\in\setN}$, $(c_j)_{j\in\setN}$, $\ppfuncSym(\cdot)$, and the
    fact that each player tries to maximize their worst-case expected payoff (as
    formalized in \cref{ def:lpg:payoff-bounds}). Players do not know $u_j$.

    Player $i$ chooses a strategy $s_i \in \setS_i$ based on the perceived game
    $(v^i_j)_{j\in\setN}$. After all players choose their strategies, the true
    game $(u_j)_{j\in\setN}$ is revealed to all players. The payoff of player
    $i$ is $u_i\sqty{s_1, \ldots, s_n}$.
\end{definition}

Based on the perceived payoff function, a player can use the following payoff
concretization function to reason about the set of possible true payoff
functions:
\begin{definition}[Payoff concretization function]
    \label{def:lpg:ppf-inv}
    With the notation in \cref{def:lpg:ppf}, let $v \in U$ be a payoff function.
    The inverse of the payoff perception function is called a \emph{payoff
    concretization function}, which is denoted as $\ppfuncSym^{-1}: \qty(U
    \times \extPosInt) \mapsto 2^U$ and defined as
    \begin{align*}
        \ppfuncInv{v}{c} \defeq \condSet{u \in U}{\ppfunc{u}{\maxq{
            \ppfMaxC{v}, c}} = v}
    \end{align*}
\end{definition}
If a player at capability level $c$ perceives $v$, then the set of possible true
payoff functions is $\ppfuncInv{v}{c}$. Note that $\ppfuncInv{v}{c}$ is
well-defined even if $c < \ppfMaxC{v}$. The payoff concretization set can be
finite or infinite. For the masked payoff game (\cref{ ex:lpg:masked-payoff}),
$\ppfuncInv[s]{\vM}{c} = \cqty{\vM}$ when $c \geq mn$. By comparison, for the
quantized payoff game (\cref{ ex:lpg:quant-payoff}), $\ppfuncInv[q]{\vM}{c}$ is
always an uncountable (thus infinite) set for $c < \infty$.

Given the above definitions, one can prove the following properties about the
payoff perception function and the payoff concretization function. Those
properties confirm an intuitive understanding of the hierarchical structure of
the perception accuracy. The proofs are deferred to the appendix (\cref{
apx:sec:ppf-property}):
\begin{enumerate}
    \item \emph{Perfect perception beyond intrinsic capability}:
        \\ For any $c \geq \ppfMaxC{u}$, we have $\ppfunc{u}{c} = u$.
    \item \emph{Perfect perception of an abstraction}:
        \\ For any $c \in \posInt$, let $v = \ppfunc{u}{c}$.
        Then $\ppfunc{v}{c} = v$. Equivalently, $\ppfMaxC{v} \leq c$.
    \item \emph{First form of information loss}:
        \\ If $\ppfunc{u}{c} \neq u$ for some $c \in \posInt$, then
        $\ppfunc{u}{c'} \neq u$ for any $c' \in \intInt{1}{c}$.
    \item \emph{Second form of information loss}:
        \\ For any $c \in \posInt$, let $v = \ppfunc{u}{c}$. If $v \neq u$,
        then $\ppfunc{v}{c'} \neq u$ for any $c' \in \posInt$.
    \item \emph{Higher capability for more accurate perception}: \\
        For $c \in \posInt$, $\ppfuncInv{v}{c+1} \subseteq \ppfuncInv{v}{c}$.
    \item \emph{Information loss expressed via concretization}: \\
        If $v' = \ppfunc{v}{c}$, then for $c' \geq c$, $\ppfuncInv{v}{c'}
        \subseteq \ppfuncInv{v'}{c}$.
\end{enumerate}

\nopagebreak

\section{Equilibrium in one-shot limited-perception games}
\label{sec:lpg:fully-capability-aware}

\subsection{Formalizing the Nash equilibrium}
This section formalizes the Nash equilibrium for \oslpGame{}s. As defined in
\cref{def:lpg:gameplay}, each player in \aoslpGame{} maximizes their worst-case
payoff. We first formalize the bounds of expected payoffs and analyze their
properties.
\begin{definition}[Bounds of expected true payoff]
    \label{def:lpg:payoff-bounds}
    In \aoslpGame{} as defined in \cref{def:lpg:gameplay}, let
    $(\vsig_i)_{j\in\setN}$ be a mixed strategy profile with $\vsig_i \in
    \simplex{\setS_i}$. If a player perceives their payoff function as $v$ and
    has a capability level of $c$, then the lower and upper bounds of the
    expected true payoffs are denoted as
    \begin{align}
    \label{eqn:lpg:payoff-bounds-def:gen}
    \begin{split}
        \minPayoff[\ppfuncSym]{v, c, \sqty{\vsig_i}_{j\in\setN}} &\defeq
            \inf_{u \in \ppfuncInv{v}{c}} {u\sqty{\vsig_1, \ldots, \vsig_n}} \\
        \maxPayoff[\ppfuncSym]{v, c, \sqty{\vsig_i}_{j\in\setN}} &\defeq
            \sup_{u \in \ppfuncInv{v}{c}} {u\sqty{\vsig_1, \ldots, \vsig_n}}
    \end{split}
    \end{align}
    We omit $\ppfuncSym$ if it is clear from the context. Of note, for a mixed
    strategy profile $\sqty{\vsig_i}_{i\in\setN}$, $u\sqty{ \vsig_1, \ldots,
    \vsig_n}$ denotes the expected payoff with respect to the payoff function
    $u$. When we focus on a single player $i$, we use
    $\minPayoffSym\lpgPlayerParam{i}$ and $\maxPayoffSym\lpgPlayerParam{i}$ to
    denote their bounds with $\vsig_{-i}$ denoting the mixed strategy profile of
    all players except $i$.
\end{definition}

\begin{theorem}[Analytical properties of the payoff bounds]
    \label{thm:lpg:payoff-bounds-convexity}
    With the notation in \cref{ def:lpg:payoff-bounds},
    $\minPayoffSym\lpgPlayerParam{i}$ is concave in $\vsig_i$, and
    $\maxPayoffSym\lpgPlayerParam{i}$ is convex in $\vsig_i$. Both functions are
    bounded and continuous in $\sqty{\vsig_1, \ldots, \vsig_n}$.
\end{theorem}
\begin{proof}
    Due to the perfect perception with infinite capability property, $u \in
    \ppfuncInv{ u}{c}$ for any $u \in U$ and $c \in \extPosInt$. Thus
    $\ppfuncInv{\cdot}{\cdot}$ is nonempty and the functions in
    \cref{eqn:lpg:payoff-bounds-def:gen} are well-defined. Given
    \cref{eqn:lpg:payoff-bounds-def:gen}, the continuity follows from the fact
    that $u\sqty{\vsig_1, \ldots, \vsig_n}$ is continuous in $\vsig_i$ for $i
    \in \setN$, and the point-wise infimum or supremum of a family of continuous
    functions is continuous \citep[Theorem 7.4.9]{ berberian2010general}. Since
    $u\sqty{\vsig_1, \ldots, \vsig_n}$ is a multilinear function, it is both
    convex and concave in~$\vsig_i$. The convexity/concavity of
    $\minPayoff{\cdot}$ and $\maxPayoff{\cdot}$ follows from the fact that the
    point-wise supremum (resp. infimum) of a family of convex (resp. concave)
    functions is convex (resp. concave) \citep[Chapter IV, Proposition
    2.1.2]{hiriart1996convexI}. The boundedness follows from the bounded
    concretization property.
\end{proof}

Since $\minPayoffSym\lpgPlayerParam{i}$ is concave, maximizing the worst payoff
has global optima; the related equilibrium also exists. By contrast, maximizing
the best payoff is a convex maximization problem, which may not have associated
equilibrium solution.

Next, we define the Nash equilibrium of the two-player \oslpGame{}. As discussed
in \cref{ sec:introduction}, the lower-capability player reasons about all
compatible higher-capability player's perceptions. It is helpful to define this
set. We call it the \emph{narrow concretization set}:
\begin{definition}[Narrow concretization set]
    \label{def:lpg:narrow-concretization}
    In \aoslpGame{} instance (\cref{def:lpg:gameplay}), assume a
    capability-$c_1$ player perceives a payoff function as $v$. Then the set of
    possible perceptions of $v$ by another capability-$c_2$ player is called the
    \emph{narrow concretization set}, defined as:
    \begin{align}
        \label{eqn:lpg:narrow-concretization}
        \ncSet{v, c_1, c_2} &\defeq \condSet{u \in U}{
            \begin{aligned}
                & \ppfMaxC{u} \leq c_2 \text{ and } \\
                & \ppfunc{v}{\minq{c_1, c_2}} = \ppfunc{u}{\minq{c_1, c_2}}
            \end{aligned}
        }
    \end{align}
\end{definition}
One can verify that the narrow concretization set $\ncSet{v, c_1, c_2}$
satisfies the following properties; the proof is deferred to the appendix
(\cref{ thm:lpg:narrow-concretization}):
\begin{itemize}
    \item If $\ppfMaxC{v} \leq c_1 < c_2$, then $\ncSet{v, c_1, c_2}
        \subseteq \ppfuncInv{v}{c_1}$.
    \item If $c_1 \geq c_2$, then $\ncSet{v, c_1, c_2} = \cqty{
        \ppfunc{v}{c_2}}$.
\end{itemize}

\begin{definition}[Nash equilibrium of two-player \oslpGame{}]
    \label{def:lpg:2p-fully-capability-aware-nash}
    In a two-player \oslpGame{} instance (\cref{ def:lpg:gameplay}) with the
    capability profile $\sqty{c_1, c_2}$, assume without loss of generality that
    $c_1 \leq c_2$. Let $U \defeq \real^{\setS_1 \times \setS_2}$ be the set of
    all possible payoff functions. Let $\sqty{u_1, v_1}$ and $\sqty{u_2, v_2}$
    be the perceived payoff functions of the first and the second player,
    respectively. Define $T_1^2 \defeq \ncSet{v_1, c_1, c_2}$.

    A Nash equilibrium of this instance is a pair $\sqty{\vx^*, \vRy}$ where
    $\vx^* \in \simplex{\setS_1}$ is a mixed strategy of the first player, and
    $\vRy: T_1^2 \mapsto \simplex{\setS_2}$ is a response function of the second
    player. The equilibrium condition is that $\vx^*$ maximizes the first
    player's worst-case payoff over perceptions in $T_1^2$, and $\vRy$ encodes a
    best response of the second player to any possibility in $T_1^2$. Formally,
    $\sqty{\vx^*, \vRy}$ is a Nash equilibrium if and only if the following
    conditions hold:
    \begin{itemize}
        \item For a perceived payoff functions $v_2' \in T_1^2$, the second
            player has no incentive to deviate from the strategy $\vRyf{v_2'}$
            given that the first player plays $\vx^*$:
            \begin{align}
            \label{eqn:lpg:2p-fully-capability-aware-nash:R2}
                \forall v_2' \in T_1^2:\:
                \vRyf{v_2'} \in
                    \argmaxS_{\vy \in \simplex{\setS_2}}
                    \minPayoff{v_2', c_2, \vx^*, \vy}
            \end{align}
        \item Let $f_{\vRy}(\vx)$ be the first player's worst payoff when they
            play $\vx$ against $\vRy$. Then $\vx^*$ is a maximizer of $f_{
            \vRy}(\vx)$:
            \begin{align}
            \label{eqn:lpg:2p-fully-capability-aware-nash:f1}
            \begin{split}
                \vx^* &\in \argmaxS_{\vx \in \simplex{\setS_1}} f_{\vRy}(\vx) \\
                \text{where } f_{\vRy}(\vx) &\defeq
                    \inf_{v_2' \in T_1^2}
                    \minPayoff{u_1, c_1, \vx, \vRyf{v_2'}}
            \end{split}
            \end{align}
    \end{itemize}
    Given a Nash equilibrium $\sqty{\vx^*, \vRy}$, the first player sample a
    pure action according to $\vx^*$, and the second player samples a pure
    action according to $\vRyf{v_2}$. The values $f_{\vRy}(\vx^*)$ and
    $\minPayoff{ v_2, c_2, \vx^*, \vRyf{v_2}}$ are the \emph{Nash values} of the
    first and second player at the equilibrium, respectively.
\end{definition}
\begin{remark}
    In \cref{ def:lpg:2p-fully-capability-aware-nash}, since $T_1^2$ only
    depends on $v_1$, $c_1$ and $c_2$, and the second player can reconstruct the
    first player's perception (i.e., $u_1 = \ppfunc{u_2}{c_1}$ and $v_1 =
    \ppfunc{v_2}{c_1}$), both players can solve $\sqty{\vx^*, \vRy}$ using their
    private information. Both players have no incentive to deviate from the
    equilibrium.
\end{remark}
\begin{remark}
    In the special case of $c_1 = c_2$, both players have no uncertainty about
    their opponent; i.e., $T_1^2 = \cqty{v_1} = \cqty{v_2}$ (if $v_2' \in
    \ncSet{ v_1, c_1, c_2}$, then \cref{ eqn:lpg:narrow-concretization} requires
    $v_1 = \ppfunc{v_2'}{c_1}$ and $\ppfMaxC{v_2'} \leq c_1$, which implies
    $v_2' = \ppfunc{v_2'}{\ppfMaxC{v_2'}} = \ppfunc{v_2'}{c_1} = v_1 = v_2$).
    The equilibrium can be represented by a pair of mixed strategies
    $\sqty{\vx^*, \vy^*}$ so that $\vRyf{v_2} = \vy^*$, compatible with the
    standard Nash equilibrium.
\end{remark}

In the multiple-player case, we can define a response function for each player
and impose constraints analogous to \cref{
eqn:lpg:2p-fully-capability-aware-nash:R2}.

We briefly explain why the Nash equilibrium exists. Since $\minPayoff{ \cdot}$
is continuous and concave in each player's strategy, the correspondence (a
correspondence is a set-valued function) from other players' response functions
to the set of a player's best response functions is therefore upper
hemicontinuous and convex-valued. Thus by considering an infinite dimensional
space that contains the Cartesian product of all players' response function
spaces, the existence of the Nash equilibrium can be proved with the generalized
Kakutani fixed-point theorem \citep{ fan1952fixed}. The formal definition of the
Nash equilibrium in the multiple-player case and the existence proof is deferred
to the appendix (\cref{ apx:sec:nash-equil}).

\subsection{Hardness of equilibrium solving and representation}
\label{sec:lpg:fully-capability-aware-hardness}

One can expect that solving a Nash equilibrium in the form of \cref{
def:lpg:2p-fully-capability-aware-nash} is computationally hard. In fact,
verifying an equilibrium is hard, which can be \NP-hard or undecidable depending
on the assumptions. For a problem $P$ whose solutions can be certified, we can
encode its solution space into a set $S_P \subset U$, and define a computable
function $\ppfunc{u}{1} = u_0$ for some marker $u_0 \in U$ if and only if $u \in
S_P$. With carefully chosen payoff values of $u_0$ and $S_P$, we can check
whether $S_P$ is empty by checking the value of $f_{\vRy}(\vx_0)$. For example,
take $P$ as the halting problem. Set $u_0(1, 1) = 2$. Let $\ppfunc{u}{1} = u_0$
if and only if $u(1, 1) = \frac{1}{k}$ for some $k \in \posInt$ and $P$ halts
after $k$ steps; otherwise, set $\ppfunc{u}{1} = u$. Then $\ppfunc{u}{1}$ is
computable, but $\card{\ppfuncInv{u_0}{1}} > 1$ is undecidable. Whether
$f_{\vRy}(\vx_0) \leq 1$ is equivalent to whether $P$ halts. Of note, the
hardness still holds even if $\minPayoff{u, c, \vx, \vy}$ and $\max_{\vy \in
\simplex{\setS_2}} \minPayoff{u, c, \vx, \vy}$ are both computable in polynomial
time for any $u$ and $c$ since the payoff values can be engineered so that $S_P$
only affects $f_{\vRy}(\cdot)$ but not the lower payoff bounds. The formal
hardness theorem and proof are deferred to the appendix (\cref{
thm:lpg:fully-capability-aware-hardness}).

Let's turn our attention to the representation of the equilibrium. Assume there
is a trusted oracle that finds one Nash equilibrium. If both players are
Turing-complete, can the oracle communicate the equilibrium to them? The answer
is not obvious since representing $\vRyf{\cdot}$ may require an infinitely sized
lookup table. Representation size is also related to the hardness of equilibrium
solving. Assume $T_1^2 \defeq \ncSet{v_1, c_1, c_2}$ is finite, enumerable, but
exponentially large. Without an efficient equilibrium representation, an attempt
to solve the equilibrium is unlikely in \PSPACE.

The following theorem provides a positive result on representing $\epsilon$-Nash
equilibria. It shows that if there is a polynomial-time algorithm that solves a
convex program involving the payoff bounds (\cref{
eqn:lpg:fully-capability-aware-nash-repr:orcl} below), then there is an
approximate Nash equilibrium $\sqty{\vx^*, \vRy}$ whose representation size does
not depend on the size of $T_1^2$, and evaluating $\vRyf{u}$ for any $u \in
T_1^2$ is also in polynomial time. Consequently, finding an approximate Nash
equilibrium can be in \PSPACE{} assuming $T_1^2$ is finite and enumerable by
brute-force search of all possible formulations compatible with \cref{
eqn:lpg:fully-capability-aware-nash-repr} below.

\begingroup
\newcommand{\vRyfp}[1]{\vRyf[']{#1}}
\newcommand{\orcl}[1]{\mathfrak{O}\sqty{#1}}
\newcommand{\fval}{f^{*}}
\begin{restatable}[Compact representation of approximate Nash
    equilibria]{theorem}{ThmCompactRepr}
    \label{thm:lpg:fully-capability-aware-nash-repr}
    In an two-player \oslpGame{} instance (\cref{ def:lpg:gameplay}), assume
    without loss of generality $c_1 \leq c_2$. Let $\sqty{u_1, v_1}$ and
    $\sqty{u_2, v_2}$ be the perceived payoff functions of the players. Let
    $\sqty{\vx^*, \vRy}$ be an $\epsilon$-Nash equilibrium of this instance with
    $\epsilon \geq 0$ as defined in \cref{
    def:lpg:2p-fully-capability-aware-nash}. Define $T_1^2 \defeq \ncSet{v_1,
    c_1, c_2}$. For $u \in T_1^2$ and $t \in \real$, define an oracle $\orcl{u,
    t}$ that returns an arbitrary solution $\vy$ subject to the following convex
    constraints:
    \begin{align}
    \label{eqn:lpg:fully-capability-aware-nash-repr:orcl}
    \begin{split}
        \orcl{u, t} &\defeq \vy \\
        \text{subject to }\quad
        & \vy \in \simplex{\setS_2},\quad
            \minPayoff{u_1, c_1, \vx^*, \vy} \geq t, \\
        & \minPayoff{u, c_2, \vx^*, \vy} \geq
            \max_{\vy' \in \simplex{\setS_2}}
            \minPayoff{u, c_2, \vx^*, \vy'} - \epsilon - \epsilon' \\
    \end{split}
    \end{align}

    Take an arbitrary value $\epsilon' > 0$. Then there exists a function
    $\vRy': T_1^2 \mapsto \simplex{\setS_2}$ such that $\sqty{\vx^*, \vRy'}$ is
    an $\qty(\epsilon + \epsilon')$-Nash equilibrium of this instance, where
    $\vRyfp{\cdot}$ is defined as the following with game-dependent parameters
    $v_i^* \in T_1^2$ and $\vy_i^* \in \simplex{\setS_2}$ index by $i \in \setR$
    for $\setR \defeq \intInt{1}{\card{\setS_1}+1}$, and $\fval \in \real$:
    \begin{align}
        \label{eqn:lpg:fully-capability-aware-nash-repr}
        \vRyfp{u} \defeq \begin{dcases}
            \vy^*_i & \text{if } u = v_i^* \text{ for some } i \in \setR \\
            \orcl{u, f^*} & \text{if } u \in T_1^2 \setminus
                \condSet{v_i^*}{i \in \setR}
        \end{dcases}
    \end{align}
    Of note, the representation size of $\vRyfp{\cdot}$ can be implicitly
    affected by $\epsilon'$ since higher precision may be needed to represent
    $v_i^*$ for smaller $\epsilon'$. In particular, if $\ppfuncInv{u}{c}$ is
    always finite, then one can set $\epsilon' = 0$.
\end{restatable}
\newcommand{\SY}{S_{\vRy}}
\newcommand{\fY}[1]{f_{\vRy}\sqty{#1}}
\begin{proofidea}
    The first player's Nash objective function \cref{
    eqn:lpg:2p-fully-capability-aware-nash:f1} can be rewritten as $\fY{\vx} =
    \min_{\vs \in \SY}\T{\vs}\vx$ where $\SY$ is a convex, compact set derived
    from the second player's response function $\vRyf{\cdot}$. With the minimax
    theorem, we can find $\vs^* \in \SY$ so that $\max_{\vx \in
    \simplex{\setS_1}} \T[*]{\vs}\vx = \max_{\vx \in \simplex{\setS_1}}
    \fY{\vx}$. Carath\'{e}odory's theorem allows us to rewrite $\vs^*$ as a
    convex combination of at most $\card{\setS_1} + 1$ vertices of $\SY$. Since
    the infimum over $T_1^2$ is involved in $\fY{\vx}$, we need to use
    approximate vertices that correspond to $v_i^*$ and $\vy_i^*$ in \cref{
    eqn:lpg:fully-capability-aware-nash-repr}. The complete proof is deferred to
    the appendix (\cref{ apx:sec:repr-proof}).
\end{proofidea}
\endgroup

\subsection{The zero-sum case}
This section defines zero-sum games in the one-shot limited-perception setting.
For an arbitrary payoff perception function, even if the true game is zero-sum,
the perceived game may be nonzero-sum. To address this issue, we require that
the payoff perception function be \emph{odd} so that the perceived game is also
zero-sum.
\begin{definition}[Odd payoff perception functions]
    \label{def:lpg:odd}
    A payoff perception function $\ppfuncSym: U \times \extPosInt \to U$ is called
    \emph{odd} if and only if
    \begin{align}
        \label{eqn:lpg:odd}
        \forall u \in U,\quad \forall c \in \posInt,\quad
        \ppfunc{-u}{c} = -\ppfunc{u}{c}
    \end{align}
\end{definition}

It is straightforward to verify that the payoff perception functions of the
masked payoff game (\cref{ex:lpg:masked-payoff}) and the quantized payoff game
(\cref{ex:lpg:quant-payoff}) are odd. It's trivial to verify the following
properties of an odd $\ppfuncSym(\cdot)$:
\begin{itemize}
    \item $\ppfMaxC{-u} = \ppfMaxC{u}$ for any $u \in U$.
    \item $\ppfuncInv{-u}{c} = \condSet{-v}{v \in \ppfuncInv{u}{c}}$ for any
        $u \in U$ and $c \in \posInt$.
    \item For $v \in U$, $c \in \posInt$, and $\vsig_k \in
        \simplex{\setS_k}$, the payoff bounds satisfy
        \begin{align*}
        \begin{split}
            \minPayoff{-v, c, \qty(\vsig_k)_{k\in\setN}} &=
                -\maxPayoff{v, c, \qty(\vsig_k)_{k\in\setN}} \\
            \maxPayoff{-v, c, \qty(\vsig_k)_{k\in\setN}} &=
                -\minPayoff{v, c, \qty(\vsig_k)_{k\in\setN}}
        \end{split}
        \end{align*}
\end{itemize}

\begin{definition}[\tfzGame{} and its notation]
    \label{def:lpg:tfzg}
    In a two-player zero-sum one-shot limited-perception game (\tfzGame{} for
    short), besides requirements given in \cref{def:lpg:gameplay}, there are
    three additional conditions:
    \begin{itemize}
        \item The payoff perception function is odd (\cref{def:lpg:odd}).
        \item Players know that the true game is zero-sum.
        \item The above two conditions are common knowledge.
    \end{itemize}
    In a \tfzGame{} instance, we use $\sqty{u_1, -u_1}$ and $\sqty{u_2, -u_2}$
    to denote the perceived payoff functions of the first and the second player,
    respectively. The first player maximizes $\minPayoff{u_1, c_1, \vx, \vy}$.
    The second player's goal is  maximizing $\minPayoff{-u_2, c_2, \vx, \vy}$,
    which is equivalent to minimizing $\maxPayoff{u_2, c_2, \vx, \vy}$.
\end{definition}

Similar to the general-sum Nash equilibrium (\cref{
def:lpg:2p-fully-capability-aware-nash}), in a zero-sum game, the
lower-capability player reasons over all possible true payoff functions and the
opponent's corresponding best response to determine their strategy. It differs
from the general-sum case in that a possible perception of the higher-capability
player is tied to the lower-capability player's true payoff function. The
following definition formalizes the zero-sum Nash equilibrium.

\begin{definition}[Nash equilibrium of \tfzGame{}]
    \label{def:lpg:tfzg-nash}
    In a \tfzGame{} instance (\cref{ def:lpg:tfzg}), assume without loss of
    generality that $c_1 \leq c_2$. Let $\sqty{u_1, -u_1}$ and $\sqty{u_2,
    -u_2}$ be the perceived payoff functions of the first and the second player,
    respectively. Define $T_1^2 \defeq \ncSet{u_1, c_1, c_2}$ where
    $\ncSet{\cdot}$ is the narrow concretization set defined in \cref{
    def:lpg:narrow-concretization}.

    A pair $\sqty{\vx^*, \vRy}$ is a Nash equilibrium of this instance if and
    only if the following conditions hold:
    \begin{itemize}
        \item $\vx^* \in \simplex{\setS_1}$ is a mixed strategy for the first
            player, and $\vRy: T_1^2 \mapsto \simplex{\setS_2}$ is the response
            function of the second player.
        \item The second player has no incentive to deviate from $\vRyf{\cdot}$
            given $\vx^*$:
            \begin{align}
            \label{eqn:lpg:tfzg-nash:R2}
                \forall u_2' \in T_1^2:\:
                \vRyf{u_2'} \in \argminS_{\vy \in \simplex{\setS_2}}
                    \maxPayoff{u_2', c_2, \vx^*, \vy}
            \end{align}
        \item Let $f_{\vRy}(\vx)$ be the first player's worst payoff when they
            play $\vx$ against $\vRy$. Then $\vx^*$ is a maximizer of $f_{
            \vRy}(\vx)$:
            \begin{align}
            \label{eqn:lpg:tfzg-nash:f1}
            \begin{split}
                \vx^* &\in \argmaxS_{\vx \in \simplex{\setS_1}} f_{\vRy}(\vx) \\
                \text{where } f_{\vRy}(\vx) &\defeq
                    \inf_{u_2' \in T_1^2}
                    \minPayoff{u_2', c_2, \vx, \vRyf{u_2'}}
            \end{split}
            \end{align}
    \end{itemize}
    Given a Nash equilibrium $\sqty{\vx^*, \vRy}$, the first player samples an
    action according to $\vx^*$, and the second player samples an action
    according to $\vRyf{u_2}$. The values $f_{\vRy}\sqty{\vx^*}$ and
    $\maxPayoff{u_2, c_2, \vx^*, \vRyf{u_2}}$ are the \emph{Nash values} of the
    first and second player at the equilibrium, respectively. Of note, the
    second player tries to minimize their Nash value.
\end{definition}

In sharp contrast to the conventional normal-form games where the zero-sum case
can be solved efficiently, finding a Nash equilibrium of a \tfzGame{} instance
is not easier than solving a general-sum game. The following theorem presents a
reduction from a general-sum instance to a \tfzGame{} instance. The key idea is
that in a \tfzGame{} instance, players are optimizing weakly coupled objectives
(one maximizing the lower bound, the other minimizing the upper bound of a
different perception). Since any game's equilibrium is invariant with respect to
affine transformations of the payoff functions, given a two-player general-sum
instance, we can separate the ranges of the players' payoffs by applying affine
transformation so that they are optimizing equivalent objectives in a
corresponding \tfzGame{} instance.
\begingroup
\newcommand{\uAz}{\bar{u}_1}
\newcommand{\uBz}{\bar{u}_2}
\newcommand{\ppfuncSymz}{\bar{\ppfuncSym}}
\newcommand{\ppfuncz}[3][]{\ppfuncSymz_{#1}\sqty{#2, #3}}
\newcommand{\vxnz}{\overbar{\vx}^*}
\newcommand{\vRynz}{\overbar{\vRy}}
\newcommand{\vRynzf}[1]{\overbar{\vRy}\sqty{#1}}
\begin{theorem}[Reducing a general-sum game to \tfzGame]
    \label{thm:lpg:tfzg-gsum-reduction}
    Let $\sqty{u_1, v_1}$ and $\sqty{u_2, v_2}$ be the perceived payoff
    functions of the players in a two-player \oslpGame{} instance with
    capability levels $c_1 \leq c_2$ and $\ppfuncSym(\cdot)$ as the payoff
    perception function. Let $ m \defeq \card{\setS_1}$ and $n \defeq
    \card{\setS_2}$ be the numbers of pure actions available to each player. Let
    $T_1^2 \defeq \ncSet{v_1, c_1, c_2}$.

    Then there exist constants $k_n, b_n \in \real$ with $k_n \neq 0$, and a
    \tfzGame{} instance where $\bar{\setS}_1 = \intInt{1}{2m}$, $\bar{\setS}_2 =
    \setS_2$, $\ppfuncSymz(\cdot)$ is the payoff perception function, and
    players have capability levels $\bar{c_1} = 1$ and $\bar{c_2} = 2$ with
    perceptions as $\sqty{\uAz, -\uAz}$ and $\sqty{\uBz, -\uBz}$, such that for
    any Nash equilibrium $\sqty{\vxnz, \vRynz}$ of the \tfzGame{} instance,
    $\sqty{\vx^*, \vRy}$ is a Nash equilibrium of the original game where
    \begin{align}
        \label{eqn:lpg:tfzg-gsum-reduction}
        \vx^* \defeq \vxnz_{:m},\quad
        \vRyf{v} \defeq \vRynzf{\mqty[\V{1} \\ k_n \vv + b_n]}
        \text{ for } v \in T_1^2
    \end{align}
    Note that we treat a payoff function and a payoff matrix interchangeably,
    with the first player being the row player. For a (row) vector or a matrix
    $\vA$, $\vA_{:m}$ denotes the first $m$ rows of $\vA$, and $\vA_{m+1:}$
    denotes the remaining rows.
\end{theorem}
\begin{proof}
    Let $U \defeq \real^{m \times n}$ and $\bar{U} \defeq \real^{2m \times n}$
    be the sets of all possible payoff functions in the general-sum and zero-sum
    games, respectively. For a set $C \subsetneq U$ of payoff matrices, use
    $\min C$ to denote the minimum payoff among the items in $C$, defined as
    $\min C \defeq \inf_{\vA \in C} \min_{i,\, j} A_{ij}$. Define $\max C$
    analogously. For $k, b \in \real$, use $kC + b$ to denote the set
    $\condSet{k\vA + b}{\vA \in S}$.

    Let $C_u \defeq \ppfuncInv{u_1}{c_1}$ and $C_v \defeq \ppfuncInv{v_1}{c_1}$.
    Note that $T_1^2 \subseteq C_v$. Due to the bounded concretization property,
    there are constants $k_n < 0 < k_p$ (where ``n'' stands for negative and
    ``p'' for positive) and $b_n, b_p \in \real$ such that
    \begin{gather*}
        0 < \min \qty(k_p C_u + b_p) \leq \max \qty(k_p C_u + b_p) < 1 \\
        1 < \min \qty(k_n C_v + b_n) \leq \max \qty(k_n C_v + b_n) < 2
    \end{gather*}
    Define the function $\uAz(\cdot)$ as the all-one function, i.e.,
    $\uAz\sqty{i, j} = 1$ for $i \in \bar{\setS}_1$ and $j \in \bar{\setS}_2$.
    Define a function $\ppfuncSymz_{0}: \bar{U} \times \extPosInt \to \bar{U}$
    as $\ppfuncz[0]{\uAz}{c} = \uAz$, and $\ppfuncz[0]{u}{c} = u$ except for the
    following cases:
    \begin{alignat*}{3}
        \forall \vV \in T_1^2:&&
        \ppfuncz[0]{\mqty[ k_n \vV' + b_n \\ \V{0} ]}{2} &=
        \mqty[\V{1} \\ k_n \vV + b_n] &\quad&
            \text{for } \vV' \in \ppfuncInv{\vV}{c_2} \\
        \forall \vV \in T_1^2:&&
        \ppfuncz[0]{\mqty[ k_p \vU + b_p \\ -(k_n \vV + b_n) ]}{2} &=
        \mqty[\V{1} \\ k_n \vV + b_n] &&
            \text{for } \vU \in C_u \\
        \forall \vV \in T_1^2:&&
        \ppfuncz[0]{\mqty[ \V{1} \\ k_n \vV + b_n ]}{2} &= \mqty[
            \V{1} \\ k_n \vV + b_n ],\:
        \mathrlap{\ppfuncz[0]{\mqty[ \V{1} \\ k_n \vV + b_n ]}{1} = \uAz} &&
    \end{alignat*}
    Define $\ppfuncSymz: \bar{U} \times \extPosInt \to \bar{U}$ as
    $\ppfuncz{-u}{c} \defeq -\ppfuncz[0]{u}{c}$ if $u$ corresponds to any
    special case above; otherwise, $\ppfuncz{u}{c} \defeq \ppfuncz[0]{u}{c}$.
    Then $\ppfuncSymz(\cdot)$ is an odd payoff perception function. Let
    $\bar{T}_1^2 \defeq \ncSet{\uAz, 1, 2}$ be the narrow concretization set
    with respect to $\ppfuncSymz( \cdot)$. We have $\bar{T}_1^2 = \cqty{\uAz}
    \cup \condSet{\mqty[\V{1} \\ k_n \vV + b_n]}{\vV \in T_1^2}$. One can verify
    that for any $u' \in \bar{T}_1^2 \setminus \cqty{\uAz}$, $\vx \in
    \simplex{2m}$, and $\vy \in \simplex{n}$, the payoff bounds satisfy:
    \newcommand{\condEq}{\overset{?}{=}}
    {\small
    \begin{gather*}
    \begin{aligned}
        \minPayoff[\ppfuncSymz]{u', 2, \vx, \vy} &=
            \inf_{\vU \in C_u}
            \T{\vx} \mqty[k_p \vU + b_p \\ -(k_n \vV_0 + b_n)] \vy
            \condEq k_p \minPayoff[\ppfuncSym]{u_1, c_1, \vx_{:m}, \vy} + b_p \\
        \maxPayoff[\ppfuncSymz]{u', 2, \vx, \vy} &\condEq
            \sup_{\vV' \in \ppfuncInv{\vV_0}{c_2}}
            \T{\vx} \mqty[k_n \vV' + b_n \\ \V{0}] \vy
            = k_n \minPayoff[\ppfuncSym]{\vV_0, c_2, \vx_{:m}, \vy} + b_n
    \end{aligned} \\
        \text{where } u' =  \mqty[\V{1} \\ k_n \vV_0 + b_n]
        \text{ for some } \vV_0 \in T_1^2, \text{ and }
        \condEq \text{ holds if } \vx_{m+1:} = 0
    \end{gather*}}%
    We also have $\minPayoff[\ppfuncSymz]{\uAz, 2, \vx, \vy} =
    \maxPayoff[\ppfuncSymz]{\uAz, 2, \vx, \vy} = 1$ for any $\vx \in
    \simplex{2m}$ and $\vy \in \simplex{n}$. Since $-(k_n \vV_0 + b_n) < 0 < k_p
    \vU + b_p < 1$ for $\vV_0 \in T_1^2$ and $\vU \in C_u$, and $\vxnz$
    maximizes $\vx \mapsto \inf_{u' \in \bar{T}_1^2} \minPayoff[\ppfuncSymz]{
    u', 2, \vx, \vRynzf{u'}}$, we have $\vxnz_{m+1:}=\V{0}$. Given \cref{
    eqn:lpg:tfzg-nash:f1, eqn:lpg:tfzg-nash:R2}, the equilibrium condition for
    $\sqty{\vxnz, \vRynz}$ is:
    {\small
    \begin{gather*}
        \vxnz \in
            \argmaxS_{\vx \in \simplex{2m}}
                \inf_{u' \in \bar{T}_1^2}
                \minPayoff[\ppfuncSymz]{u', 2, \vx, \vRynzf{u'}}
            = \argmaxS_{\substack{\vx \in \simplex{2m} \\
                    \vx_{m+1:} = \V{0}}}
                \inf_{\substack{u' \in \bar{T}_1^2 \\
                        u' \neq \uAz}}
                \minPayoff[\ppfuncSym]{u_1, c_1, \vx_{:m}, \vRynzf{u'}} \\
        \forall u' \in \bar{T}_1^2 \setminus \cqty{\uAz}:\:
        \vRynzf{u'} \in
            \argminS_{\vy \in \simplex{n}}
                \maxPayoff[\ppfuncSymz]{u', 2, \vxnz, \vy}
            = \argmaxS_{\vy \in \simplex{n}}
                \minPayoff[\ppfuncSym]{\vV_{u'}, c_2, \vxnz_{:m}, \vy} \\
        \text{where } \vV_{u'} \defeq \frac{\vu'_{m+1:} - b_n}{k_n}
        \text{ so that }
        \mqty[\V{1} \\ k_n \vV_{u'} + b_n] = u'
        \text{ always holds}
    \end{gather*}}%
    Define $\sqty{\vx^*, \vRy}$ as in \cref{ eqn:lpg:tfzg-gsum-reduction}.
    Then the above conditions are equivalent to \cref{
    eqn:lpg:2p-fully-capability-aware-nash:f1,
    eqn:lpg:2p-fully-capability-aware-nash:R2} for $\sqty{\vx^*, \vRy}$ in the
    original game.
    \qed
\end{proof}
\endgroup

\section{Efficiently solvable zero-sum games}
\label{sec:lpg:efficient-zero-sum}

\subsection{The maximin-attainable \tfzGame{} instances}
A conventional zero-sum game can be efficiently solved due to the maximin
theorem. We generalize the condition to \tfzGame{}s by defining
\emph{maximin-attainable} instances. In all \tfzGame{} instances (\cref{
def:lpg:tfzg}), the lower-capability player's Nash value lies between their
maximin value and the Stackelberg value. If the maximin value equals the
Stackelberg value, the instance is called maximin-attainable.

\begin{definition}[Stackelberg objective in \tfzGame{}]
    \label{def:lpg:tfzg-stackelberg}
    In a \tfzGame{} instance as defined in \cref{ def:lpg:tfzg} with $c_1 \leq
    c_2$, define $T_1^2 \defeq \ncSet{u_1, c_1, c_2}$. For a second player's
    response function $\vRy: T_1^2 \times \simplex{\setS_1} \mapsto
    \simplex{\setS_2}$, define the first player's \emph{Stackelberg objective
    function} $\gRy: \simplex{\setS_1} \mapsto \real$ as
    \begin{align}
    \label{eqn:lpg:tfzg-stackelberg:g}
    \begin{split}
        \gRyf{\vx} &\defeq \inf_{u_2' \in T_1^2} \minPayoff{u_2', c_2, \vx,
            \vRyf{u_2', \vx}}
    \end{split}
    \end{align}
    A response function $\vRyf{\cdot}$ is called a \emph{best response function}
    if and only if:
    \begin{align}
    \label{eqn:lpg:tfzg-stackelberg:R2}
        \forall \sqty{u_2', \vx} \in T_1^2 \times \simplex{\setS_1}:\:
        \Ryf{u_2', \vx} \in \argminS_{\vy \in \simplex{\setS_2}}
            \maxPayoff{u_2', c_2, \vx, \vy}
    \end{align}
    If $\gRyf{\cdot}$ is a best response function, the corresponding
    \emph{Stackelberg value} of the first player is defined as $V_s \defeq
    \sup_{\vx \in \simplex{\setS_1}} \gRyf{\vx}$.
\end{definition}

\begin{lemma}
    \label{thm:lpg:tfzg-p1-for-all}
    In a two-player \oslpGame{} instance where $c_1 \leq c_2$, let $\Ry: U
    \times \simplex{\setS_1} \mapsto 2^{\simplex{\setS_2}}$ be a
    non-empty-valued set-valued function. Define $T_1^2 \defeq \ncSet{u_1, c_1,
    c_2}$.

    For any $\vx \in \simplex{\setS_1}$, it holds that
    \begin{align}
    \label{eqn:lpg:tfzg-p1-for-all}
        \inf_{u' \in T_1^2} \inf_{\vy \in \Ryf{u', \vx}}
        \minPayoff{u', c_2, \vx, \vy}
        = \inf_{u' \in \ppfuncInv{u_1}{c_1}}
        \inf_{\vy \in \Ryf{ \ppfunc{u'}{c_2}, \vx}} u'\sqty{\vx, \vy}
    \end{align}
\end{lemma}
\begin{proof}
    Recall that $\minPayoff{u_1, c_1, \vx, \vy} = \inf_{u' \in
    \ppfuncInv{u_1}{c_1}}u'\sqty{\vx, \vy}$ (\cref{ def:lpg:payoff-bounds}) and
    $T_1^2 = \condSet{u' \in U}{ \ppfMaxC{u'} \leq c_2
    \text{ and } \ppfunc{u'}{c_1} = u_1 }$ (\cref{
    def:lpg:narrow-concretization}). Define a set
    $
        S \defeq \condSet{u'' \in U}{ \exists u' \in
        T_1^2:\: u'' \in \ppfuncInv{u'}{c_2}}
    $. We have
    \begin{align*}
        &\inf_{u' \in T_1^2} \inf_{\vy \in \Ryf{u', \vx}}
            \minPayoff{u', c_2, \vx, \vy} =
        \inf_{u' \in T_1^2} \inf_{\vy \in \Ryf{u', \vx}}
            \inf_{u'' \in \ppfuncInv{u'}{c_2}} u''\sqty{\vx, \vy} \\
        &\quad =
        \inf_{u' \in T_1^2} \inf_{u'' \in \ppfuncInv{u'}{c_2}}
        \inf_{\vy \in \Ryf{\ppfunc{u''}{c_2}, \vx}} u''\sqty{\vx, \vy}
        = \inf_{u'' \in S} \inf_{\vy \in \Ryf{\ppfunc{u''}{c_2}, \vx}}
            u''\sqty{\vx, \vy}
    \end{align*}
    To prove \cref{eqn:lpg:tfzg-p1-for-all}, it suffices to show that $S =
    \ppfuncInv{u_1}{c_1}$. By \cref{def:lpg:ppf}, $\ppfuncInv{u_1}{c_1} =
    \condSet{u \in U}{\ppfunc{u}{c_1} = u_1}$. For $u' \in \ppfuncInv{u_1}{
    c_1}$, if $\ppfMaxC{u'} \leq c_2$, then $u' \in T_1^2$ and $u' \in
    \ppfuncInv{u'}{c_2}$; otherwise, let $u'' \defeq \ppfunc{u'}{c_2}$; then
    $u'' \in T_1^2$ and $u' \in \ppfuncInv{u''}{c_2}$. Thus $u' \in S$ in both
    cases. Conversely, for $u' \in S$, there exists $u''$ such that $u_1 =
    \ppfunc{u''}{c_1}$ and $u'' = \ppfunc{u'}{c_2}$, which implies $u_1 =
    \ppfunc{u'}{c_1}$ and thus $u' \in \ppfuncInv{u_1}{c_1}$. Therefore, $S =
    \ppfuncInv{u_1}{c_1}$.
\end{proof}

\begingroup
\newcommand{\vRyp}{\vRy^*}
\newcommand{\vRypf}[1]{\vRyf[^*]{#1}}
\newcommand{\vRynf}[1]{\vRyf[_n]{#1}}
\begin{lemma}[Order of \tfzGame{} equilibrium values]
    \label{thm:lpg:tfzg-equil-order}
    In a \tfzGame{} instance as defined in \cref{ def:lpg:tfzg} with $c_1 \leq
    c_2$, let $T_1^2 \defeq \ncSet{u_1, c_1, c_2}$.

    Define $h_1: \simplex{\setS_1} \mapsto \real$ as the first player's
    \emph{maximin objective function}:
    \begin{align}
        \label{eqn:lpg:tfzg-h1}
        h_1(\vx) \defeq \inf_{\vy \in \simplex{\setS_2}}
            \minPayoff{u_1, c_1, \vx, \vy}
    \end{align}
    Define $V_h \defeq \sup_{\vx \in \simplex{\setS_1}} h_1(\vx)$ as the
    \emph{maximin value}.

    Let $\sqty{\vx^*_n, \vRy_n}$ be a Nash equilibrium, and $V_n$ be the
    corresponding Nash value of the first player (\cref{ def:lpg:tfzg-nash}).

    Given $\vRynf{\cdot}$, define a response function $\vRyp: T_1^2 \times
    \simplex{\setS_1} \mapsto \simplex{\setS_2}$ as:
    \begin{align*}
        \vRypf{u, \vx} &\defeq \begin{dcases}
            \vRynf{u} & \text{if } \vx = \vx^*_n \\
            \argmin_{\vy \in \simplex{\setS_2}} \maxPayoff{u, c_2, \vx, \vy}
                & \text{otherwise}
        \end{dcases}
    \end{align*}
    Then $\vRypf{\cdot}$ is a best response function; let $V_s$ be the
    corresponding Stackelberg value (\cref{ def:lpg:tfzg-stackelberg}). Then the
    following inequality holds:
    \begin{align}
    \label{eqn:lpg:tfzg-equil-order}
        V_h \leq V_n \leq V_s
    \end{align}

    Moreover, if $V_h = V_s$, then $\sqty{\vx^*_n, \vRynf{u_2}}$ is a
    Stackelberg equilibrium.
\end{lemma}
\begin{proof}
    \Cref{ eqn:lpg:tfzg-nash:R2} implies that $\vRypf{\cdot}$ satisfies \cref{
    eqn:lpg:tfzg-stackelberg:R2}. Define $f_1: \simplex{\setS_1}^2 \mapsto
    \real$ as:
    \[
        f_1\sqty{\vx, \vx'} \defeq
            \inf_{u' \in T_1^2} \minPayoff{u', c_2, \vx, \vRypf{u', \vx'}}
    \]
    Plugging \cref{ def:lpg:payoff-bounds} and \cref{ thm:lpg:tfzg-p1-for-all}
    into $h_1(\cdot)$ yields
    \begin{align*}
        \forall \vx' \in \simplex{\setS_1}:\:
        h_1(\vx) = \inf_{u' \in T_1^2} \inf_{\vy \in \simplex{\setS_2}}
                \minPayoff{u', c_2, \vx, \vy} \leq
                f_1\sqty{\vx, \vx'}
    \end{align*}
    It's easy to verify that the Nash value of the first player as defined in
    \cref{ def:lpg:tfzg-nash} can be redefined as
    $
        V_n =  f_1\sqty{\vx^*_n, \vx^*_n} = \sup_{\vx \in \simplex{\setS_1}}
        f_1\sqty{\vx, \vx^*_n}
    $.
    Then $V_h = \sup_{\vx \in \simplex{\setS_1}} h_1(\vx) \leq \sup_{\vx \in
    \simplex{\setS_1}} f_1\sqty{\vx, \vx^*_n} = V_n$.

    We have $\gRyf{\vx} = \inf_{u' \in T_1^2} \minPayoff{u', c_2, \vx, \vRyf{u',
    \vx}} = f_1\sqty{\vx, \vx}$, which yields $V_s = \sup_{\vx} \gRyf{\vx} =
    \sup_{\vx} f_1\sqty{\vx, \vx} \geq f_1\sqty{\vx^*_n, \vx^*_n} = V_n$.
    Therefore, \cref{ eqn:lpg:tfzg-equil-order} holds. It also follows that
    $\sqty{\vx^*_n, \vRynf{u_2}}$ is a Stackelberg equilibrium when $V_h = V_s$.
    \qed
\end{proof}
\endgroup

Given \cref{thm:lpg:tfzg-equil-order}, it is natural to consider games with $V_h
= V_s$ so that the lower-capability player's Nash value equals their maximin
value. The following definition captures this concept.
\begin{definition}[Maximin-attainable \tfzGame{} instances]
    \label{def:lpg:maximin-attainable}
    In a \tfzGame{} instance as defined in \cref{ def:lpg:tfzg} with $c_1 \leq
    c_2$, let $T_1^2 \defeq \ncSet{u_1, c_1, c_2}$. Given any best response
    function $\vRy: T_1^2 \times \simplex{\setS_1} \mapsto \simplex{\setS_2}$
    satisfying \cref{ eqn:lpg:tfzg-stackelberg:R2}, let $\gRyf{\cdot}$ and $V_s$
    be the corresponding Stackelberg objective function and value as defined in
    \cref{ def:lpg:tfzg-stackelberg}. Let $V_h$ and $h_1(\cdot)$ be the first
    player's maximin value and maximin objective function, respectively, as
    defined in \cref{ thm:lpg:tfzg-equil-order}. The instance is
    \emph{maximin-attainable} for the first player if and only if $V_h =
    \sup_{\vx \in \simplex{ \setS_1}} \gRyf{\vx}$ for any $\vRyf{\cdot}$
    satisfying~\cref{ eqn:lpg:tfzg-stackelberg:R2}.
\end{definition}

\begingroup
\newcommand{\vxapx}{\vx_{\operatorname{approx}}}
\renewcommand{\fs}[2]{\mathfrak{f}_{#1}\sqty{#2}}
Since $\minPayoff{\cdot}$ is a concave function, the infimum in the maximin
objective function can be replaced by finite minimization over the extreme
points of $\simplex{\setS_2}$, which allows tractable computation:
\begin{theorem}[Equilibrium of maximin-attainable \tfzGame{} instances]
    \label{thm:lpg:maximin-equil}
    Assume a \tfzGame{} instance with $c_1 \leq c_2$ is maximin-attainable for
    the first player as defined in \cref{ def:lpg:maximin-attainable}. There
    is a strategy profile $\sqty{\vx^*, \vy^*}$ such that:
    \begin{align}
    \label{eqn:lpg:maximin-equil}
    \begin{gathered}
        \vx^* \in \argmaxS_{\vx \in \simplex{\setS_1}} \fs1{\vx},\quad
        \vy^* \in \argminS_{\vy \in \simplex{\setS_2}} \fs2{\vy} \\
        \text{where }
            \fs1{\vx} \defeq
                \min_{1 \leq k \leq \card{\setS_2}}
                \minPayoff{u_1, c_1, \vx, \ve_k},\quad
            \fs2{\vy} \defeq \maxPayoff{u_2, c_2, \vx^*, \vy}
    \end{gathered}
    \end{align}
    Any strategy profile $\sqty{\vx^*, \vy^*}$ satisfying \cref{
    eqn:lpg:maximin-equil} is a Stackelberg equilibrium when the first player is
    the leader. For the first player, all Nash values are the same as the
    Stackelberg value. Furthermore, there exists $\vRy_n: T_1^2 \mapsto
    \simplex{\setS_2}$ such that $\sqty{\vx^*, \vRy_n}$ is a Nash equilibrium as
    defined in \cref{ def:lpg:tfzg-nash}.
\end{theorem}
\begin{proof}
    The existence of $\vx^*$ and $\vy^*$ is due to compactness of $\simplex{
    \setS_1}$ and $\simplex{\setS_2}$, and the continuity of $\fs1{\cdot}$ and
    $\fs2{\cdot}$, which follows from the continuity of payoff bounds (\cref{
    thm:lpg:payoff-bounds-convexity}).

    Let $h_1(\cdot)$ be the maximin objective function of the first player as
    defined in \cref{ eqn:lpg:tfzg-h1}. Since $\vy \mapsto \minPayoff{u_2, c_2,
    \vx, \vy}$ is concave and continuous, Bauer Maximum Principle \citep[Theorem
    7.69]{ aliprantis2013infinite} implies that
    \[
        \fs1{\vx} = \min_{1 \leq k \leq \card{\setS_2}}
            \minPayoff{u_1, c_1, \vx, \ve_k}
            = \inf_{\vy \in \simplex{\setS_2}} \minPayoff{u_1, c_1, \vx, \vy}
            = h_1(\vx)
    \]

    To prove the last paragraph of the theorem, note that the definition of the
    maximin-attainable \tfzGame{} requires that $h_1(\vx^*) = \sup_{\vx \in
    \simplex{\setS_1}} \gRyf{\vx}$ for any best response function
    $\vRyf{\cdot}$. For any $\vx' \in \simplex{\setS_1}$, we have
    \begin{align*}
        \gRyf{\vx'} &\leq \sup_{\vx \in \simplex{\setS_1}} \gRyf{\vx}
        = h_1(\vx^*)
        = \inf_{u' \in \ncSet{u_1,c_1,c_2}} \inf_{\vy \in \simplex{\setS_2}}
            \minPayoff{u', c_2, \vx^*, \vy} \\
        &\leq \inf_{u' \in \ncSet{u_1,c_1,c_2}}
            \minPayoff{u', c_2, \vx^*, \vRyf{u', \vx^*}}
        = \gRyf{\vx^*}
    \end{align*}
    Therefore, $\vx^* \in \argmaxS_{\vx \in \simplex{\setS_1}} \gRyf{\vx}$,
    which means $\sqty{\vx^*, \vy^*}$ is a Stackelberg equilibrium. \Cref{
    thm:lpg:tfzg-equil-order} implies that $V_h = V_n = V_s$ for any Nash value
    $V_n$. The proof of the existence of $\vRy_n$ is deferred to the appendix
    (\cref{ apx:thm:lpg:maximin-equil-nash}).
    \qed
\end{proof}

In a maximin-attainable \tfzGame{} instance, the lower-capability player should
play their maximin strategy $\vx^*$. The higher-capability player solves $\vx^*$
from their private information and plays a best response to it. Although the
complete response function of the higher-capability player that forms an
equilibrium with $\vx^*$ is not necessarily solvable, the lower-capability
player has no reason to deviate given \cref{ thm:lpg:maximin-equil}.
Furthermore, if computing $\minPayoff{u, c, \vx, \vy}$ and its subgradient is in
polynomial time, then the objective function $\fs1{\cdot}$ defined in \cref{
thm:lpg:maximin-equil} and its subgradient is computable in polynomial time.
Therefore, an approximate maximin strategy can be efficiently computed using
nonsmooth convex optimization methods~\citep{ bubeck2015convex, jia2024trafs}.

\subsection{Some maximin-attainable games}
\label{sec:lpg:maximin-attainable-examples}

A simple example of a maximin-attainable game is the \emph{constant-gap games}.
If there exists a function $\delta: \posInt \to \real$ such that $\maxPayoff{u,
c, \vx, \vy} = \minPayoff{u, c, \vx, \vy} + \delta(c)$ for all $u \in U$, $c \in
\posInt$, $\vx \in \simplex{\setS_1}$, and $\vy \in \simplex{\setS_2}$, then the
game is called a constant-gap game. It is easy to verify that the quantized
payoff game (\cref{ ex:lpg:quant-payoff}) is a constant-gap game with $\delta(c)
= 10^{-c}$.

\begin{theorem}
    \label{thm:lpg:constant-gap-maximin}
    A constant-gap \tfzGame{} instance is maximin-attainable.
\end{theorem}
\begin{proof}
    Let $\vRyf{\cdot}$ be a best response function defined in \cref{
    eqn:lpg:tfzg-stackelberg:R2}.
    \begin{align*}
    \begin{gathered}
        \forall u \in \ncSet{u_1, c_1, c_2}:\:
        \forall \vx \in \simplex{\setS_1}: \\
        \begin{aligned}
        &\minPayoff{u, c_2, \vx, \vRyf{u, \vx}}
            = \maxPayoff{u, c_2, \vx, \vRyf{u, \vx}} - \delta(c_2) \\
            &\qquad= \min_{\vy \in \simplex{\setS_2}}
                \maxPayoff{u, c_2, \vx, \vy} - \delta(c_2)
            = \min_{\vy \in \simplex{\setS_2}} \minPayoff{u, c_2, \vx, \vy}
        \end{aligned}
    \end{gathered}
    \end{align*}
    Consequently, $h_1(\vx) = \gRyf{\vx}$ holds for all $\vx \in
    \simplex{\setS_1}$ with $h_1(\cdot)$ and $\gRyf{\cdot}$ defined in \cref{
    eqn:lpg:tfzg-h1} and \cref{eqn:lpg:tfzg-stackelberg:g}, respectively.
\end{proof}

A more interesting example is the \emph{narrowly reversible games}. Note that
for all games, $\minPayoff{u, c, \vx, \vy} = \inf_{u' \in \ppfuncInv{u}{c}}
\maxPayoff{u, c, \vx, \vy}$. A game is narrowly reversible if the analogous
equality holds for $\ncSet{u, c, c'}$ instead of $\ppfuncInv{u}{c}$.

\begin{definition}[Narrowly reversible payoff games]
    \label{def:lpg:narrowly-reversible}
    In a \tfzGame{} instance as defined in \cref{ def:lpg:tfzg} with $c_1 \leq
    c_2$, let $T_1^2 \defeq \ncSet{u_1, c_1, c_2}$. The game is \emph{narrowly
    reversible} for the first player if and only if the following holds:
    \begin{align}
    \label{eqn:lpg:lower-narrowly-reversible}
    \forall \vx \in \simplex{ \setS_1}:\:
    \forall k \in \card{\setS_2}:\:
        \minPayoff{u_1, c_1, \vx, \ve_k} = \inf_{u' \in T_1^2}
        \maxPayoff{u', c_2, \vx, \ve_k}
    \end{align}
\end{definition}
\begin{theorem}
    \label{thm:lpg:narrowly-reversible}
    A narrowly reversible \tfzGame{} instance is maximin-attainable.
\end{theorem}
\begin{proof}
    Let $V_h$ and $h_1(\cdot)$ be the first player's maximin value and objective
    function as defined in \cref{ thm:lpg:tfzg-equil-order}. Let $\vRyf{\cdot}$
    be a best response function satisfying \cref{ eqn:lpg:tfzg-stackelberg:R2}.
    Let $\gRyf{\cdot}$ be the corresponding Stackelberg objective function.
    Take any $\vx \in \simplex{\setS_1}$. Let $i \in \intInt{1}{\card{\setS_2}}$
    such that $h_1(\vx) = \minPayoff{u_1, c_1, \vx, \ve_i}$. Take an arbitrary
    value $\epsilon \in \posReal$. Due to \cref{
    eqn:lpg:lower-narrowly-reversible}, there
    exists $u' \in T_1^2$ such that $\maxPayoff{u', c_2, \vx,
    \ve_i} \leq \minPayoff{u_1, c_1, \vx, \ve_i} + \epsilon = h_1(\vx) +
    \epsilon$. Let $\vy_{u'} = \vRyf{u', \vx}$. It holds that
    \begin{align*}
        \gRyf{\vx} &\leqN1 \minPayoff{u', c_2, \vx, \vy_{u'}}
            \leq \maxPayoff{u', c_2, \vx, \vy_{u'}}
            \leqN2 \maxPayoff{u', c_2, \vx, \ve_i} \\
            &\leq h_1(\vx) + \epsilon,
    \end{align*}
    where $\leqN1$ follows from taking $u'$ in the infimum in the definition of
    $\gRyf{\cdot}$ (see \cref{ eqn:lpg:tfzg-stackelberg:g}), and $\leqN2$ is due
    to $\vRyf{u', \vx}$ minimizing $\vy \mapsto \maxPayoff{u', c_2, \vx, \vy}$.
    Taking $\epsilon \to 0$ yields $\gRyf{\vx} \leq h_1(\vx)$. Since $h_1(\vx)
    \leq \gRyf{\vx}$ (see the proof of \cref{ thm:lpg:tfzg-equil-order}), it
    follows that $h_1(\vx) = \gRyf{ \vx}$.
\end{proof}

In an instance of the masked payoff game (\cref{ ex:lpg:masked-payoff}), let
$\vA \in \real^{m\times n}$ be the perceived matrix of the first (i.e., row)
player. If for each column of $\vA$, there are at most $(c_2 - c_1)$ masked
entries, then the game is narrowly reversible since $\minPayoff{\vA, c_1, \vx,
\ve_k} = \maxPayoff{\vA', c_2, \vx, \ve_k}$ where the $k$-th column of $\vA'$ is
the $k$-th column of $\vA$ with the masked entries filled with the minimum
possible value.

The limited-rank game outlined in \cref{ ex:lpg:limited-rank} is always narrowly
reversible. We state the following theorem. The proof is deferred to \cref{
apx:sec:limited-rank-games}. Of note, \cref{ eqn:lpg:lrk-payoff-bounds} implies
that the maximin strategy and the best response in a limited-rank game can both
be solved by Second-Order Cone Programs (SOCPs). An $\epsilon$-approximate
solution to a SOCP can be found in $O(-\log \epsilon)$ iterations using
interior-point methods \citep[Chapter 11]{ boyd2004convex}.
\begin{theorem}[Payoff bounds of limited-rank games]
    Let $\vA \in \real^{m \times n}$ be the perceived matrix of a capability-$c$
    player in a limited-rank game. Then for any $\vx \in \simplex{\setS_1}$ and
    $\vy \in \simplex{\setS_2}$, the payoff bounds are given by:
    \begin{align}
    \label{eqn:lpg:lrk-payoff-bounds}
    \begin{split}
        \minPayoff{\vA, c, \vx, \vy} &= \T{\vx} \vA \vy
            - \uncPayoff{\vA, c, \vx, \vy} \\
        \maxPayoff{\vA, c, \vx, \vy} &= \T{\vx} \vA \vy
            + \uncPayoff{\vA, c, \vx, \vy} \\
        \text{where } \uncPayoff{\vA, c, \vx, \vy}
            &\defeq \sigma_c(\vA) \cdot
            \norm{\leftNullT{\vA}\vx}_2 \cdot \norm{\NullT{\vA}\vy}_2
    \end{split}
    \end{align}
    Moreover, the limited-rank game is narrowly reversible. Of note,
    $\Null{\vM}$ is the null space of a matrix $\vM$, and $\sigma_c(\vM)$ is the
    $c$-th largest singular value of $\vM$.
\end{theorem}

\renewcommand{\bibname}{References}
\bibliographystyle{splncs04}
\bibliography{refs.bib}

\clearpage
\appendix

\section{Properties of the payoff perception function}
\label{apx:sec:ppf-property}

\begin{theorem}[Properties of the payoff perception function]
    \label{thm:lpg:ppf-property}
    With the notation in \cref{def:lpg:ppf}, let $u \in U$ be a payoff
    perception function. Then the following properties hold:
    \begin{itemize}
        \item \emph{Perfect perception beyond intrinsic capability}:
            For any $c \geq \ppfMaxC{u}$, we have $\ppfunc{u}{c} = u$.
        \item \emph{Perfect perception of an abstraction}:
            \\ For any $c \in \posInt$, let $v = \ppfunc{u}{c}$.
            Then $\ppfunc{v}{c} = v$. Equivalently, $\ppfMaxC{v} \leq c$.
        \item \emph{First form of information loss}:
            \\ If $\ppfunc{u}{c} \neq u$ for some $c \in \posInt$, then
            $\ppfunc{u}{c'} \neq u$ for any $c' \in \intInt{1}{c}$.
        \item \emph{Second form of information loss}:
            \\ For any $c \in \posInt$, let $v = \ppfunc{u}{c}$. If $v \neq u$,
            then $\ppfunc{v}{c'} \neq u$ for any $c' \in \posInt$.
    \end{itemize}
\end{theorem}
\begin{proof}
    To prove the first two properties, one can substitute $u$ (resp. $v$) with
    $\ppfunc{u}{\ppfMaxC{u}}$ (resp. $\ppfunc{v}{c}$) and apply the capability
    path independence property:
    \begin{align*}
        \ppfunc{u}{c} &= \ppfunc{\ppfunc{u}{\ppfMaxC{u}}}{c} =
        \ppfunc{u}{\minq{\ppfMaxC{u}, c}} = \ppfunc{u}{\ppfMaxC{u}} = u \\
        \ppfunc{v}{c} &= \ppfunc{\ppfunc{u}{c}}{c} =
        \ppfunc{u}{\minq{c, c}} = \ppfunc{u}{c} = v
    \end{align*}

    The first form of information loss is a direct consequence of the first
    property, since $\ppfunc{u}{c'} = u$ implies $\ppfunc{u}{c} = u$ for $c \geq
    c'$.

    To prove the second form of information loss, note that $\ppfunc{v}{c'} =
    \ppfunc{u}{\minq{c, c'}}$. We have $\ppfunc{v}{c'} = \ppfunc{u}{c} = v \neq
    u$ when $c' > c$, and $\ppfunc{v}{c'} = \ppfunc{u}{c'} \neq u$ when $c' \leq
    c$ (due to the first form of information loss).
\end{proof}

\begin{theorem}[Abstraction reduces perception accuracy]
    \label{thm:lpg:abs-hierarchy}
    With the notation in \cref{def:lpg:ppf,def:lpg:ppf-inv}, for any $v \in U$
    and $c \in \integer$, it holds that
    \begin{align}
        \ppfuncInv{v}{c+1} \subseteq \ppfuncInv{v}{c}
        \label{eqn:lpg:abs-hierarchy:1}
    \end{align}

    Let $v' = \ppfunc{v}{c}$, then for $c' \geq c$, it holds that
    \begin{align}
        \ppfuncInv{v}{c'} \subseteq \ppfuncInv{v'}{c}
        \label{eqn:lpg:abs-hierarchy:2}
    \end{align}
\end{theorem}
\begin{proof}
    To prove \cref{eqn:lpg:abs-hierarchy:1}, let $u \in \ppfuncInv{v}{c+1}$,
    which means $\ppfunc{u}{\maxq{\ppfMaxC{v}, c+1}} = v$. If $c \geq
    \ppfMaxC{v}$, then $\ppfunc{u}{c+1} = v$. Applying the properties of
    capability path independence property and perfect perception beyond
    intrinsic capability yields $\ppfunc{u}{c} = \ppfunc{\ppfunc{u}{c+1}}{c} =
    \ppfunc{v}{c} = v$. If $c < \ppfMaxC{v}$, then $ \ppfunc{u}{\maxq{
    \ppfMaxC{v}, c}} = \ppfunc{u}{\maxq{\ppfMaxC{v}, c+1}} = v$. In both
    cases, we have $\ppfunc{u}{\maxq{\ppfMaxC{v}, c}} = v$, which means $u \in
    \ppfuncInv{v}{c}$.

    To prove \cref{eqn:lpg:abs-hierarchy:2}, first consider the case when $c
    \geq \ppfMaxC{v}$, which implies $v'= v$ due to the property of perfect
    perception beyond intrinsic capability. Then \cref{eqn:lpg:abs-hierarchy:2}
    follows from \cref{eqn:lpg:abs-hierarchy:1} by induction on $c'$.
    Now assume $c < \ppfMaxC{v}$. Let $u \in \ppfuncInv{v}{c'}$,
    which means $\ppfunc{u}{\maxq{\ppfMaxC{v}, c'}} = v$. We also have
    $\ppfMaxC{v'} \leq c$ due to the property of perfect perception of an
    abstraction. Then $v' = \ppfunc{v}{c} =
    \ppfunc{\ppfunc{u}{\maxq{\ppfMaxC{v}, c'}}}{c} =
    \ppfunc{u}{\minq{\maxq{\ppfMaxC{v}, c'}, c}} = \ppfunc{u}{c}$. We thus have
    $\ppfunc{u}{\maxq{c, \ppfMaxC{v'}}} = \ppfunc{u}{c} = v'$, which means $u
    \in \ppfuncInv{v'}{c}$.
\end{proof}

An interpretation of \cref{ eqn:lpg:abs-hierarchy:1} is that given the same
perceived payoff function, if a player has a higher capability, then they can
have a more precise set of possible true payoff functions. \Cref{
eqn:lpg:abs-hierarchy:2} can be understood as another form of information loss,
as the set of possible true payoff functions becomes larger when an abstraction
is applied to a perceived payoff function.

\begin{theorem}[Property of the narrow concretization set]
    \label{thm:lpg:narrow-concretization}
    The narrow concretization set $\ncSet{v, c_1, c_2}$ satisfies the following
    properties:
    \begin{itemize}
        \item If $\ppfMaxC{v} \leq c_1 < c_2$, then $\ncSet{v, c_1, c_2}
            \subseteq \ppfuncInv{v}{c_1}$.
        \item If $c_1 \geq c_2$, then $\ncSet{v, c_1, c_2} = \cqty{
            \ppfunc{v}{c_2}}$.
        \item If $c_1' \leq c_1$, then $\ncSet{v, c_1, c_2} \subseteq
            \ncSet{\ppfunc{v}{c_1'}, c_1', c_2}$.
    \end{itemize}
\end{theorem}
\begin{proof}
    If $\ppfMaxC{v} \leq c_1 < c_2$ and $u \in \ncSet{v, c_1, c_2}$, then
    $\ppfunc{u}{c_1} = \ppfunc{v}{c_1} = v$, which means $u \in
    \ppfuncInv{v}{c_1}$.

    If $c_1 \geq c_2$ and $u \in \ncSet{v, c_1, c_2}$, then $ u \eqN1
    \ppfunc{u}{c_2} = \ppfunc{v}{c_2}$, where $\eqN1$ is due to $\ppfMaxC{u}
    \leq c_2$ and \cref{thm:lpg:ppf-property}.

    To prove the third property, consider a few cases. If $c_1 \leq c_2$, let $u
    \in \ncSet{v, c_1, c_2}$. Then $\ppfunc{u}{c_1} = \ppfunc{v}{c_1}$ and
    $\ppfMaxC{u} \leq c_2$. Given $c_1' \leq c_1$, the capability path
    independence implies that $\ppfunc{u}{c_1'} = \ppfunc{\ppfunc{u}{c_1}}{c_1'}
    = \ppfunc{\ppfunc{v}{c_1}}{c_1'} = \ppfunc{\ppfunc{v}{c_1'}}{c_1'}$, which
    yields $u \in \ncSet{\ppfunc{v}{ c_1'}, c_1', c_2}$. If $c_1 > c_2$, then
    $\ncSet{v, c_1, c_2} = \cqty{w}$ where $w \defeq \ppfunc{v}{c_2}$. If $c_2
    \leq c_1' \leq c_1$, then $\ncSet{\ppfunc{v}{c_1'}, c_1', c_2} = \cqty{
    \ppfunc{ \ppfunc{v}{c_1'} }{ c_2}} = \cqty{w}$. If $c_1' < c_2 < c_1$, then
    $\ppfunc{w}{\minq{c_1', c_2}} = \ppfunc{v}{c_1'} = \ppfunc{\ppfunc{v}{
    c_1'}}{ \minq{c_1', c_2}}$, implying $w \in \ncSet{\ppfunc{v}{c_1'}, c_1',
    c_2}$. Combining all cases yields $\ncSet{v, c_1, c_2} \subseteq \ncSet{
    \ppfunc{ v}{ c_1'}, c_1', c_2}$.
\end{proof}


\section{Nash equilibrium in multiple-player \oslpGame{}}
\label{apx:sec:nash-equil}
\begin{definition}[Nash equilibrium of games]
    \label{def:lpg:fully-capability-aware-nash}
    In an instance of \oslpGame{} as defined in \cref{ def:lpg:gameplay}, define
    $T_i^j \defeq \ncSet{v_j^i, c_i, c_j}$ as the set of possible perceived
    payoff functions of player $j$ from the perspective of player $i$. Let $m =
    \argmin_{i \in \setN} c_i$ be a player with the lowest capability level. For
    $i \in \setN$, let $\vR_i: T_m^i \mapsto \simplex{\setS_i}$ be the response
    function of player $i$. The tuple $\sqty{\vR_1, \ldots, \vR_n}$ is called a
    \emph{response profile}.

    A response profile is a Nash equilibrium of the instance if and only if the
    following condition holds for each $i \in \setN$:
    \begin{align}
    \label{eqn:lpg:fully-capability-aware-nash:obj}
    \begin{split}
        \forall u \in T_m^i:\:
        \vR_i(u) &\in \argmaxS_{\vsig \in \simplex{\setS_i}}
            f_i\lrmid{u, \vsig}{\vR_{-i}} \\
        \text{where } f_i\lrmid{u, \vsig}{\vR_{-i}} &\defeq
            \inf_{v'_j \in T_i^j \text{ for } j \in N_i}
            \minPayoffSym\lrmid{u, c_i, \vsig}{
                \vR_j\sqty{v'_j} \text{ for } j \in N_i} \\
            N_i &\defeq \setN \setminus \cqty{i}
    \end{split}
    \end{align}
    At the Nash equilibrium, player $i$ samples an action from the distribution
    defined by $\vR_i(v_i^i)$. The value $f_i\lrmid{v_i^i, \vR_i(v_i^i)}{
    \vR_{-i}}$ is called the \emph{Nash value} of player $i$ at the equilibrium.
\end{definition}
\begin{remark}
    For any $i \in \setN$, since $T_j^i \subseteq T_m^i$ for all $j \in \setN$
    (third property in \cref{ thm:lpg:narrow-concretization}), $T_m^i$ is
    sufficient to contain the values of $u$ used in the infimum of other
    players' $f_j$ for $j \in \setN \setminus \cqty{i}$. Therefore, \cref{
    eqn:lpg:fully-capability-aware-nash:obj} only needs to ensure that
    $\vR_i(u)$ is a best response to $u$ and $\vR_{-i}$ for $u \in T_m^i$. Of
    note, $T_m^i$ is bounded due to the bounded concretization property (\cref{
    def:lpg:ppf}).

    In the two-player case with $c_1 \leq c_2$, the case of $i=1$ in~\cref{
    eqn:lpg:fully-capability-aware-nash:obj} corresponds to \cref{
    eqn:lpg:2p-fully-capability-aware-nash:f1}, and the case of $i=2$
    corresponds to the constraints on $\vRyf{v_2'}$ for $v_2' \in T_1^2$ in
    \cref{ eqn:lpg:2p-fully-capability-aware-nash:R2} (note that $T_1^1 = T_2^1
    = \cqty{u_1}$, $\vRy = \vR_2$, and $\vx^* = \vR_1(u_1)$). Therefore, the
    definition is consistent with \cref{ def:lpg:2p-fully-capability-aware-nash}
    in the two-player case.

    If $c_i = c_j$ for all $i, j \in \setN$, then $T_i^j = \cqty{v_j^j}$ for all
    $i, j \in \setN$. The definition is consistent with the conventional Nash
    equilibrium where players have no uncertainty about the other players'
    objective functions.
\end{remark}

\begin{theorem}[{\citep[Theorem 1]{fan1952fixed}}]
    \label{thm:fan-kakutani}
    Let $L$ be a locally convex topological vector space and $K$ a compact
    convex set in $L$. Let $\mathfrak{K}(K)$ be the family of all closed convex
    nonempty subsets of $K$. Then for any upper hemicontinuous point-to-set
    transformation $f$ from $K$ into $\mathfrak{K}(K)$, there exists a point
    $x_0 \in K$ such that $x_0 \in f(x_0)$.
\end{theorem}

\begingroup
\newcommand{\bounded}[1]{\mathcal{B}\sqty{#1}}
\begin{theorem}[Existence of Nash equilibrium in \oslpGame{}s]
    \label{thm:lpg:fully-capability-aware-nash-exist}
    The Nash equilibrium of \aoslpGame{} instance as defined in \cref{
    def:lpg:fully-capability-aware-nash} always exists.
\end{theorem}
\begin{proof}
    Define the vector space $X \defeq \prod_{i\in \setN} \bounded{U \times
    \setS_i}$ over $\real$ where $\bounded{A}$ denotes the set of bounded
    real-valued functions on $A$ and the addition and scalar multiplication are
    defined point-wise on the functions. For any point $p \in X$, use the tuple
    $p = \sqty{\vR_1, \ldots, \vR_n}$ to denote the components of $p$ where
    $\vR_i \in \bounded{U \times \setS_i}$; $\vR_i$ is interpreted as a function
    from $U$ to $\real^{ \setS_i}$. Extend $X$ to a metric space by defining the
    distance function using the uniform convergence norm defined as
    $\norm{\sqty{\vR_1, \ldots, \vR_n}}_\infty \defeq \max_{i \in \intInt{1}{n}}
    \sup_{u \in U} \max_{j \in \intInt{1}{\card{\setS_i}}} \abs{R_{ij}(u)}$
    where $R_{ij}(u)$ is the $j$-th component of $\vR_i(u)$. Of note, the
    supremum in the norm's definition is finite since $\vR_i(\cdot)$ is a
    bounded function. Since any norm is also a seminorm, $X$ is a locally convex
    topological vector space~\citep[Theorem 5.73]{aliprantis2013infinite}.
    Define the subset $S \defeq \prod_{i\in \setN} \simplex{\setS_i}^U \subseteq
    X$. It is easy to verify that $S$ is compact, convex, and nonempty.

    Given the capability profile $\sqty{c_1, \ldots, c_n}$ and the perceived
    payoff functions $\sqty{v_i^j}_{\sqty{i, j} \in \setN^2}$, with $f_i\lrmid{
    \cdot, \cdot}{\cdot}$ and $T_m^i$ defined as in \cref{
    def:lpg:fully-capability-aware-nash}, define the function $g: S^2
    \mapsto \real$ as:
    \begin{align}
    \begin{split}
        &g\sqty{\qty(\vR_i)_{i\in \setN}, \qty(\vR'_i)_{i\in \setN}} \defeq \\
        &\qquad
            \sum_{i \in \setN}
            \inf_{u \in T_m^i}
            \qty(
                f_i\lrmid{u, \vR'_i\qty(u)}{\vR_{-i}}  -
                \max_{\vsig \in \simplex{\setS_i}}
                f_i\lrmid{u, \vsig}{\vR_{-i}}
            )
    \end{split}
    \end{align}

    The function $g\sqty{\cdot, \cdot}$ is bounded since the bounded
    concretization property (\cref{ def:lpg:ppf}) implies that $T_m^i$ is
    bounded. Since $\minPayoffSym\lrmid{u, c_i, \vsig}{\vR_{j}(u'_j)}$ is
    continuous in both $\vsig$ and $\vR_{-i}$ and concave in $\vsig$ (\cref{
    thm:lpg:payoff-bounds-convexity}), $f_i\lrmid{u, \vsig}{\vR_{-i}}$ is
    continuous in both $\vsig$ and $\vR_{-i}$ and concave in $\vsig$. Therefore,
    the supremum of $\vsig \mapsto f_i\lrmid{u, \vsig}{\vR_{-i}}$ is always
    attainable. It also follows that $g(p,\, q)$ is continuous in both $p$ and
    $q$ and concave in $q$.

    Define a set-valued function $\Phi: S \mapsto 2^S$ as:
    \begin{align}
        \Phi(p) \defeq \condSet{q \in S}{g(p,\, q) = 0}
    \end{align}
    The definition of $g(\cdot,\, \cdot)$ implies that $g(p,\, q) \leq 0$. By
    setting $\vR'_i(u)$ as an arbitrary maximizer of $\vsig \mapsto f_i\lrmid{u,
    \vsig}{\vR_{-i}}$ over $\simplex{\setS_i}$, $\Phi(p)$ is nonempty. Since
    $g(p,\, q)$ is concave and continuous in $q$, $\Phi(p)$ is compact and
    convex. Moreover, since $g(p,\, q)$ is continuous in both $p$ and $q$,
    $\fGr \Phi \defeq \condSet{\sqty{p, q} \in S^2}{q \in \Phi(p)}$
    is closed, and $\Phi(\cdot)$ is thus upper hemicontinuous \citep[Theorem
    17.11]{ aliprantis2013infinite}. \Cref{thm:fan-kakutani} implies that there
    exists $p_0 \in S$ such that $p_0 \in \Phi(p_0)$. Let $p_0 = \sqty{\vR_1^*,
    \ldots, \vR_n^*}$. One can easily verify that $\sqty{\vR_1^*, \ldots,
    \vR_n^*}$ is a Nash equilibrium of the \oslpGame{} instance.
\end{proof}
\endgroup


\section{Hardness results about Nash equilibrium in two-player \oslpGame{}}
\label{apx:sec:nash-hardness}

\subsection{Hardness of computing a Nash equilibrium}
\begin{theorem}[Hardness of checking the best response]
    \label{thm:lpg:fully-capability-aware-hardness}
    Define the following decision problem in the context of an instance of a
    two-player \oslpGame{} as defined in \cref{ def:lpg:gameplay}. Assume all
    input values of the problem and all input/output values of intermediate
    Turing machines can be exactly represented as rational numbers of polynomial
    size. The problem is \NP-complete.
    \begin{itemize}
        \item \textbf{Input}:
            \begin{enuminline}
                \item $m \in \posInt$
                \item $\epsilon \in \posReal$
                \item A Turing machine that defines a payoff perception function
                    $\ppfunc{\cdot}{\cdot}$.
            \end{enuminline}
        \item \textbf{Constraints on input}:
            \begin{enuminline}
                \item The strategy spaces are $\setS_1 = \intInt{1}{m}$ and
                    $\setS_2 = \cqty{1, 2}$
                \item For any $u \in U$, the payoff perception function
                    satisfies $\ppfunc{u}{2} = u$, $\ppfunc{u}{1}$ is computable
                    in $\poly(m)$ time, $\ppfunc{ \cdot}{ \cdot}$ is
                    representable by a Turing machine of size $\poly(m)$,
                    $\card{ \ppfuncInv{ u}{ 1}}$ is finite, and
                    $\ppfuncInv{u}{1}$ is enumerable by a Turing machine of size
                    $\poly(m)$
                \item Given $u \in U$, $\vx \in \simplex{\setS_1}$, and $c \in
                    \posInt$, both $\minPayoff{u, c, \vx, \vy}$ for any $\vy \in
                    \simplex{\setS_2}$ and $\argmax_{\vy \in \simplex{\setS_2}}
                    \minPayoff{u, c, \vx, \vy}$ are computable in $\poly(m)$
                    time.
            \end{enuminline}
        \item \textbf{Constants}:
            The first player's perceived payoff function of the second player,
            denoted as $v_1$, is defined as $v_1\sqty{x, y} = -y$ for all
            $\sqty{x, y} \in \setS_1 \times \setS_2$. Define $\vy_0 = \T{\mqty[0
            & 1]}$, $c_1 = 1$, and $c_2 = 2$.
        \item \textbf{Arbitrary values}:
            \begin{enuminline}
                \item $\vx_0 \in \simplex{\setS_1}$ is arbitrarily chosen
                \item an arbitrary Turing machine that defines a best response
                    function $\vRy: U \mapsto \simplex{\setS_2}$ such that for
                    all $u \in U$, $\vRyf{u} \in \argmaxS_{\vy \in
                    \simplex{\setS_2}} \minPayoff{u, c_2, \vx_0, \vy}$, and
                    $\vRyf{u}$ is computable in $\poly(m)$ time.
            \end{enuminline}
            Of note, $\vx_0$ and $\vRyf{\cdot}$ can be chosen depending on the
            input; there is no choice of $\vx_0$ or $\vRyf{\cdot}$ that makes
            the problem easier.
        \item \textbf{Output}: whether there exists $v_2' \in \ncSet{v_1, c_1,
            c_2}$ such that $\norm{\vRyf{v_2'} - \vy_0}_\infty \leq \epsilon$.
    \end{itemize}
\end{theorem}
\begin{proof}
    The proof uses a reduction from the Boolean satisfiability problem (SAT), a
    well-known \NP-complete problem \citep{garey1979computers}. Let $\phi$ be a
    Boolean formula with $n$ variables and $\poly(n)$ clauses. Denote the
    variables in $\phi$ as $\cqty{t_1, \ldots, t_n}$. Let $m = 2n$. Define a set
    $S_\phi \subseteq U$ as a set of payoff functions that encodes the solution
    to the SAT problem $\phi$. Given a payoff function $u \in U$, $u \in S_\phi$
    if and only if all of the following conditions hold:
    \begin{itemize}
        \item $u\sqty{x, 1} = 1$ for all $x \in \intInt{1}{2n}$.
        \item $\cqty{u\sqty{2k-1, 2}, u\sqty{2k, 2}} = \cqty{2, 3}$ for all $k
            \in \intInt{1}{n}$.
        \item Let $t_k$ be true if and only if $u\sqty{2k, 2} = 3$. The
            assignment $\cqty{t_1, \ldots, t_n}$ satisfies the formula~$\phi$.
    \end{itemize}

    Let $v_1\sqty{x, y} = -y$ as defined above. Define a function $\ppfuncSym: U
    \times \posInt \to U$ as
    \begin{align}
    \label{eqn:lpg:fully-capability-aware-hardness-ppf}
    \begin{split}
        \ppfunc{u}{c} &=
            \begin{dcases}
                v_1 & \text{if } u \in S_\phi \cup \cqty{v_1}
                    \text{ and } c = 1 \\
                u & \text{otherwise}
            \end{dcases} \\
    \end{split}
    \end{align}
    It is straightforward to verify that $\ppfunc{\cdot}{\cdot}$ satisfies
    \cref{def:lpg:ppf} and is thus a payoff perception function. One can also
    easily verify that $\ppfunc{u}{c}$ is computable in $\poly(n)$ time with a
    Turing machine of size $\poly(n)$.

    The construction above further ensures that $\phi$ is satisfiable if and
    only if there exists $u \neq v_1$ such that $\ppfunc{u}{1} = v_1$.
    Thus it holds that
    \begin{align}
        \label{eqn:lpg:fully-capability-aware-hardness0}
        \ppfuncInv{v_1}{1} = S_\phi \cup \cqty{v_1}
    \end{align}

    Define $\vy(a) \defeq \T{\mqty[a & 1 - a]}$ where $0 \leq a \leq 1$. For any
    $\vx \in \simplex{\setS_1}$, we have
    \begin{align}
    \label{eqn:lpg:fully-capability-aware-hardness1}
    \begin{split}
        v_1\sqty{\vx, \vy(a)} &= (-1) \cdot a + (-2) \cdot (1-a) = a - 2 < 0 \\
        u\sqty{\vx, \vy(a)} &= a + \lambda(1-a) = (1-\lambda)a + \lambda > 0
            \quad\text{for }u \in S_\phi \\
        \text{where } \lambda &\defeq
            \sum_{k=1}^n\qty(u\sqty{2k-1, 2}x_{2k-1} + u\sqty{2k, 2}x_{2k})
    \end{split}
    \end{align}
    The value of $\lambda$ in \cref{eqn:lpg:fully-capability-aware-hardness1}
    satisfies $2 \leq \lambda \leq 3$ due to the definition of $u\sqty{\cdot,
    2}$ for $u \in S_\phi$. It follows that for $c \in \extPosInt$ and $u \in
    S_\phi$, $\minPayoff{v_1, c, \vx, \vy(a)} = a - 2$ and $\minPayoff{u, c,
    \vx, \vy(a)} = (1-\lambda)a + \lambda$. Thus
    \begin{align}
    \label{eqn:lpg:fully-capability-aware-hardness2}
    \begin{split}
        \forall c \in \extPosInt:\quad
        \argmaxS_{a \in [0,\, 1]}
            \minPayoff{v_1, c, \vx, \vy(a)} &= \cqty{1} \\
        \forall c \in \extPosInt,\, \forall u \in S_\phi:\quad
        \argmaxS_{a \in [0,\, 1]}
            \minPayoff{u, c, \vx, \vy(a)} &= \cqty{0}
    \end{split}
    \end{align}
    Therefore, $\argmax_{\vy \in \simplex{\setS_2}} \minPayoff{u, c, \vx, \vy}$
    is computable in $\poly(n)$ time for any $u \in U$, and thus the payoff
    perception function defined by \cref{
    eqn:lpg:fully-capability-aware-hardness-ppf} satisfies the constraints
    specified in \cref{ thm:lpg:fully-capability-aware-hardness}.

    Define a set $J \defeq \condSet{\vRyf{v_2'}}{v_2' \in \ncSet{v_1, c_1,
    c_2}}$. Combining \cref{ eqn:lpg:fully-capability-aware-hardness0,
    eqn:lpg:fully-capability-aware-hardness2} with the definition of $J$ yields
    \begin{align}
        \label{eqn:lpg:fully-capability-aware-hardness3}
        J = \begin{dcases}
            \cqty{\T{\mqty[1 & 0]}, \T{\mqty[0 & 1]}}
                & \text{if } \phi \text{ is satisfiable} \\
            \cqty{\T{\mqty[1 & 0]}} & \text{otherwise}
        \end{dcases}
    \end{align}
    As a result, by setting $\epsilon < 1$, whether $\vy' \in J$ for $\norm{\vy'
    - \vy_0}_\infty \leq \epsilon$ is equivalent to whether $\phi$ is
    satisfiable, which makes the problem \NP-hard. On the other hand, since we
    assume $\ppfuncInv{u}{1}$ is finite and enumerable, the problem is in \NP.
    Therefore, the problem is \NP-complete.
    \qed
\end{proof}

A consequence of \cref{ thm:lpg:fully-capability-aware-hardness} is that
verifying an $\epsilon$-Nash equilibrium or an $\epsilon$-Stackelberg
equilibrium of \aoslpGame{} instance is both \NP-hard and \coNP-hard. With the
setup in the proof of \cref{ thm:lpg:fully-capability-aware-hardness}, one can
choose a payoff function $u_1$, a strategy $\vx_0$, and a value $\epsilon \in
\posReal$ such that $\sqty{\vx_0, \vRyf{\cdot}}$ is an $\epsilon$-Nash
equilibrium if and only if $\vy_0 \in J$. Similarly, one can choose values such
that $\sqty{\vx_0, \vRyf{\cdot}}$ is an $\epsilon$-Nash equilibrium if and only
if $\vy_0 \notin J$.

\subsection{Proof to \cref{thm:lpg:fully-capability-aware-nash-repr}}
\label{apx:sec:repr-proof}

\begingroup
\newcommand{\vRyfp}[1]{\vRyf[']{#1}}
\newcommand{\orcl}[1]{\mathfrak{O}\sqty{#1}}
\newcommand{\fval}{f^{*}}
\newcommand{\vsY}{\vs_{\vRy}}
\newcommand{\SY}{S_{\vRy}}
\newcommand{\SYcore}{\SY^{\circ}}
\newcommand{\SYp}{S_{\vRy'}}
\newcommand{\fY}[1]{f_{\vRy}\sqty{#1}}
\newcommand{\fYp}[1]{f_{\vRy'}\sqty{#1}}
\newcommand{\ftv}{\bar{f}}  
\ThmCompactRepr*
\begin{proof}
    Let $u'\sqty{\cdot, \vy} \in \real^{\card{\setS_1}}$ denote the vector
    containing the expected payoffs of the first player for each pure action
    given $\vy$. Note that
    \begin{align}
    \label{eqn:lpg:fully-capability-aware-nash-repr:payoff-bounds}
    \begin{split}
        \minPayoff{u, c, \vx, \vy} &=
            \inf_{u' \in \ppfuncInv{u}{c}} u'\sqty{\vx, \vy}
            = \inf_{\vs \in \lpgCharSetP{u, c}{\vy}} \T{\vs}\vx \\
        \text{where } \lpgCharSetP{u, c}{\vy} &\defeq
            \condSet{u'\sqty{\cdot, \vy}}{u' \in \ppfuncInv{u}{c}}
    \end{split}
    \end{align}
    Given $\vRy: T_1^2 \mapsto \simplex{\setS_2}$, define $f_{\vRy}:
    \simplex{\setS_1} \mapsto \real$ as in \cref{
    eqn:lpg:2p-fully-capability-aware-nash:f1} and apply \cref{
    eqn:lpg:fully-capability-aware-nash-repr:payoff-bounds}:
    \begin{align}
    \label{eqn:lpg:fully-capability-aware-nash-repr:SY}
    \begin{split}
        \fY{\vx} &= \inf_{v_2' \in T_1^2}
            \inf_{\vs \in \lpgCharSetP{u_1, c_1}{\vRyf{v_2'}}} \T{\vs}\vx
        = \min_{\vs \in \SY}\T{\vs}\vx \\
        \text{where } \SY &\defeq \setConv \setCl
            \bigcup_{v_2' \in T_1^2} \lpgCharSetP{u_1, c_1}{\vRyf{v_2'}}
    \end{split}
    \end{align}
    The closure in $\SY$ makes it compact so the infimum can be replaced with
    minimum. The convex hull preserves the extreme value of a linear objective
    over a compact set. The set $\SY$ is bounded due to the bounded
    concretization property.

    Define $\ftv \defeq \max_{\vx \in \simplex{\setS_1}} \fY{\vx}$ as the first
    player's exact Nash value given $\vRyf{\cdot}$. With the minimax theorem,
    $\ftv$ can be rewritten as:
    \begin{align*}
    \begin{gathered}
        \ftv = \max_{\vx \in \simplex{\setS_1}} \min_{\vs \in \SY}\T{\vs}\vx
        = \min_{\vs \in \SY} \max_{\vx \in \simplex{\setS_1}} \T{\vs}\vx
        = \max_{\vx \in \simplex{\setS_1}} \T[*]{\vsY}\vx \\
        \text{where } \T[*]{\vsY} \defeq
            \argmin_{\vs \in \SY} \max_{\vx \in \simplex{\setS_1}}
                \T{\vs}\vx
    \end{gathered}
    \end{align*}
    Carath\'{e}odory's theorem states that there are points $\vsY^{(i)} \in
    \real^{\card{\setS_1}}$ indexed by $i \in \setR$ with $\setR \defeq
    \intInt{1}{\card{\setS_1}+1}$ such that $\vsY^* = \sum_{i \in \setR}
    \alpha_i \vsY^{(i)}$ for $\V{\alpha} \in \simplex{\setR}$ and $\vsY^{(i)}
    \in \setCl \bigcup_{v_2' \in T_1^2} \lpgCharSetP{u_1, c_1}{\vRyf{v_2'}}$.

    Given $\epsilon' > 0$, there exists $v_i^* \in T_1^2$ and $\vsY^{(i)\prime}
    \in \lpgCharSetP{u_1, c_1}{ \vRyf{v_i^*}}$ for $i \in \setR$ such that
    $\norm{\vsY^{(i)\prime} - \vsY^{(i)}}_2 \leq \epsilon'$. Define $\vsY'
    \defeq \sum_{i \in \setR} \alpha_i \vsY^{(i)\prime}$. Then we have
    \begin{align*}
    \begin{split}
        \forall \vx \in \simplex{\setS_1}:
        \abs{\T[\prime]{\vsY}\vx - \T[*]{\vsY}\vx}
        & \leq \sum_{i \in \setR}
            \alpha_i \abs{\T[(i)\prime]{\vsY}\vx - \T[(i)]{\vsY}\vx} \\
        & \leq \sum_{i \in \setR}
            \alpha_i \norm{\vsY^{(i)\prime} - \vsY^{(i)}}_2 \norm{\vx}_2
        \leq \sum_{i \in \setR} \alpha_i \epsilon'
        = \epsilon'
    \end{split}
    \end{align*}
    Let $\SYcore \defeq \setConv \setCl \bigcup_{i \in \setR} \lpgCharSetP{u_1,
    c_1}{ \vRyf{v_i^*}}$. Then $\vsY' \in \SYcore$ and $\SYcore \subseteq \SY$.

    The response function $\vRyfp{\cdot}$ is constructed so that it agrees with
    $\vRyf{\cdot}$ on $\SYcore$. Let $\vy^*_i \defeq \vRyf{v_i^*}$ for $i \in
    \setR$. Let $\fval \defeq \fY{\vx^*}$. For $u \in T_1^2 \setminus
    \condSet{v_i^*}{ i \in \setR}$, since $\vRyf{u}$ is a solution to \cref{
    eqn:lpg:fully-capability-aware-nash-repr:orcl} given $t = \fval$, the oracle
    $\orcl{u, \fval}$ is well-defined. Define $\vRyfp{\cdot}$ as in \cref{
    eqn:lpg:fully-capability-aware-nash-repr}. Then $\vRyfp{v_i^*} =
    \vRyf{v_i^*}$ for $i \in \setR$.  Since $\sqty{\vx^*, \vRy}$ is an
    $\epsilon$-Nash equilibrium, we have $\ftv - \epsilon \leq \fval \leq \ftv
    $. Given $\vRyfp{\cdot}$, define $\fYp{\cdot}$ and $\SYp$ as in \cref{
    eqn:lpg:fully-capability-aware-nash-repr:SY}. We then have $\SYcore
    \subseteq \SYp$, which yields
    \begin{align}
    \label{eqn:lpg:fully-capability-aware-nash-repr:fYpmax}
    \begin{split}
        \max_{\vx \in \simplex{\setS_1}} \fYp{\vx}
        &= \max_{\vx \in \simplex{\setS_1}} \min_{\vs \in \SYp} \T{\vs}\vx
        \leq \max_{\vx \in \simplex{\setS_1}} \min_{\vs \in \SYcore} \T{\vs}\vx
        \\ &\leq \max_{\vx \in \simplex{\setS_1}} \T[\prime]{\vsY}\vx
        \leq \max_{\vx \in \simplex{\setS_1}} \T[*]{\vsY}\vx + \epsilon'
        = \ftv + \epsilon'
    \end{split}
    \end{align}

    Due to the definition of $\orcl{\cdot}$ in \cref{
    eqn:lpg:fully-capability-aware-nash-repr:orcl}, we have
    \begin{align}
        \label{eqn:lpg:fully-capability-aware-nash-repr:fYpmin1}
        \forall u \in T_1^2 \setminus \condSet{v_i^*}{i \in \setR}:\:
        \minPayoff{u_1, c_1, \vx^*, \vRyfp{u}} \geq \fval \geq \ftv - \epsilon
    \end{align}
    For $i \in \setR$, let $\lpgCharSetSym_i' \defeq \lpgCharSetP{u_1,
    c_1}{\vRyfp{v_i^*}}$. Then $ \lpgCharSetSym_i' \subseteq \SYcore \subseteq
    \SY$, which implies
    \begin{align}
        \label{eqn:lpg:fully-capability-aware-nash-repr:fYpmin2}
        \minPayoff{u_1, c_1, \vx^*, \vRyfp{v_i^*}}
        = \inf_{\vs \in \lpgCharSetSym_i'} \T{\vs}\vx^*
        \geq \min_{\vs \in \SY} \T{\vs}\vx^*
        = \fval \geq \ftv - \epsilon
    \end{align}
    Combining \cref{ eqn:lpg:fully-capability-aware-nash-repr:fYpmin1,
    eqn:lpg:fully-capability-aware-nash-repr:fYpmin2} yields $\fYp{\vx^*} \geq
    \ftv - \epsilon$, which further implies that $\fYp{\vx^*} \geq \max_{\vx \in
    \simplex{\setS_1}} \fYp{\vx} - (\epsilon + \epsilon')$ when combined with
    \cref{ eqn:lpg:fully-capability-aware-nash-repr:fYpmax}. Thus $\vx^*$ is an
    $(\epsilon + \epsilon')$-strategy for the first player. The definitions of
    $\vRyfp{\cdot}$ and $\orcl{\cdot}$ ensure that $\vRyfp{\cdot}$ is an
    $(\epsilon + \epsilon')$-strategy for the second player. Therefore,
    $\sqty{\vx^*, \vRy'}$ is an $\qty(\epsilon + \epsilon')$-Nash equilibrium.
    \qed
\end{proof}
\endgroup


\section{Nash equilibrium in maximin-attainable games}
\begingroup
\newcommand{\Af}[1]{A_{u}\sqty{#1}}
\newcommand{\cond}[1]{$\mathcal{C}_{#1}$}
\begin{theorem}[Maximin strategy as Nash in maximin-attainable \tfzGame{}]
    \label{apx:thm:lpg:maximin-equil-nash}
    In \aoslpGame{} instance as defined in \cref{ def:lpg:tfzg}, let $T_1^2
    \defeq \ncSet{u_1, c_1, c_2}$. Assume the instance is maximin-attainable as
    defined in \cref{ def:lpg:maximin-attainable}. For any maximin strategy
    $\vx^* \in \simplex{\setS_1}$ of the first player, there is a response
    function $\vRy_n: T_1^2 \mapsto \simplex{\setS_2}$ such that $\sqty{\vx^*,
    \vRy_n}$ is a Nash equilibrium as defined in \cref{ def:lpg:tfzg-nash}.
\end{theorem}
\begin{proof}
    For $u \in T_1^2$, define a correspondence $A_u: \simplex{\setS_1}
    \rightrightarrows \simplex{\setS_2}$ as the best response set of the second
    player: $\Af{\vx} \defeq \argmaxS_{\vy \in \simplex{\setS_2}} \minPayoff{u,
    c_2, \vx, \vy}$.  Since $\minPayoff{u, c_2, \vx, \vy}$ is continuous and
    concave in both $\vx$ and $\vy$ (\cref{ thm:lpg:payoff-bounds-convexity}),
    it follows that $\Af{\cdot}$ is an upper hemicontinuous correspondence
    (Berge Maximum Theorem \citep[Theorem 17.31]{ aliprantis2013infinite}), and
    $\Af{ \vx}$ is convex and compact for all $\vx \in \simplex{ \setS_1}$.

    Pick an arbitrary vector $\vq \in \real^{\card{\setS_1}}$ such that $\vq
    \neq \V{0}$ and $\vx^* + \vq \in \simplex{\setS_1}$. Define the sequence
    $\cqty{\vx_n}$ as $\vx_n \defeq \vx^* + \frac{1}{n} \vq$. For each $u \in
    T_1^2$, pick an arbitrary sequence $\cqty{\vy_n^u}$ such that $\vy_n^u \in
    \Af{\vx_n}$. Then the limit $\vy^{u*} \defeq \lim_{n \to \infty} \vy_n^u$
    exists and $\vy^{u*} \in \Af{\vx^*}$ (\citep[Theorem 17.20]{
    aliprantis2013infinite}).

    Define a response function $\vRy^*: T_1^2 \times \simplex{\setS_1} \mapsto
    \simplex{\setS_2}$ as:
    \begin{align*}
        \vRyf[^*]{u, \vx} \defeq \begin{dcases}
            \vy_n^u & \text{if } \vx = \vx_n \\
            \vy^{u*} & \text{if } \vx = \vx^* \\
            \text{an arbitrary element in } \Af{\vx} & \text{otherwise}
        \end{dcases}
    \end{align*}
    It is easy to verify that $\vRyf[^*]{\cdot}$ satisfies \cref{
    eqn:lpg:tfzg-stackelberg:R2}. Define the Nash response function as
    $\vRyf[_n]{u} \defeq \vRyf[^*]{u, \vx^*}$.

\newcommand{\Gf}[1]{G_{u}\sqty{#1}}
\newcommand{\Gpf}[1]{G'_{u}\sqty{#1}}
    For any $u \in T_1^2$, define $G_u: \simplex{\setS_1}^2 \mapsto \real$ as
    \[
        \Gf{\vx, \vx'}
            \defeq \minPayoff{u, c_2, \vx, \vRyf[^*]{u, \vx'}}
    \]
\newcommand{\Ff}[1]{F\sqty{#1}}
    Define a function $F: \simplex{\setS_1}^2 \mapsto \real$ as
    \[
        \Ff{\vx, \vx'} \defeq \inf_{u \in T_1^2} \Gf{\vx, \vx'}
    \]
    Let $V \defeq \sup_{\vx \in \simplex{\setS_1}} h_1(\vx)$. Then $V =
    h_1(\vx^*) = g_{\vRy^*}(\vx^*)$ according to \cref{ thm:lpg:maximin-equil},
    where $h_1(\cdot)$ is defined in \cref{ eqn:lpg:tfzg-h1}, and $g_{\vRy^*}(
    \cdot)$ is defined as in \cref{ eqn:lpg:tfzg-stackelberg:g}.
    It is easy to verify that $\Ff{\cdot}$ satisfies the following properties:
    \begin{enumerate}
        \item $\Ff{\vx, \vx'}$ is continuous and concave in $\vx$ for any $\vx'
            \in \simplex{\setS_1}$.
        \item $\Ff{\vx, \vx} = g_{\vRy^*}(\vx)$ for all $\vx \in \simplex{
            \setS_1}$.
        \item $\Ff{\vx^*, \vx^*} = g_{\vRy^*}(\vx^*) = V$ and $\Ff{\vx, \vx}
            \leq V$ for all $\vx \in \simplex{\setS_1}$.
        \item $\sqty{\vx^*, \vRy_n}$ is a Nash equilibrium if and only if
            $\vx^* \in \argmaxS_{\vx \in \simplex{\setS_1}} \Ff{\vx, \vx^*}$.
        \item Plugging \cref{ thm:lpg:tfzg-p1-for-all} into the definition of
            $h_1(\cdot)$ yields that:
            \begin{align*}
                \forall \vx' \in \simplex{\setS_1}:
                \Ff{\vx^*, \vx^*}
                &= h_1(\vx^*)
                = \inf_{u \in T_1^2} \inf_{\vy \in \simplex{\setS_2}}
                    \minPayoff{u, c_2, \vx^*, \vy} \\
                &\leq \inf_{u \in T_1^2} \Gf{\vx^*, \vx'}
                = \Ff{\vx^*, \vx'}
            \end{align*}
    \end{enumerate}

    \newcommand{\vxb}{\bar{\vx}}
    Now we prove by contradiction. Assume that $\sqty{\vx^*, \vRy_n}$ is not a
    Nash equilibrium. Then there exists $\delta > 0$ such that $\max_{\vx \in
    \simplex{\setS_1}} \Ff{\vx, \vx^*} \geq V + 2\delta$. Define the set $\setM
    \defeq \condSet{\vx \in \simplex{\setS_1}}{\Ff{\vx, \vx^*} \geq V +
    \delta}$. Since $\Ff{\vx, \vx^*}$ is continuous and concave in $\vx$,
    $\setM$ is compact, convex, and nonempty. There also exists $\epsilon > 0$
    such that $B_{\epsilon}[\vx^*] \subseteq \setM$, where $B_{\epsilon}[\vx^*]
    \defeq \condSet{\vx \in \simplex{\setS_1}}{\norm{\vx - \vx^*}_2 \leq
    \epsilon}$ is the closed $\epsilon$-ball centered at $\vx^*$.

    Let $\vd \defeq \frac{\epsilon}{\norm{\vq}_2} \vq$ and $\vxb \defeq \vx^* +
    \vd$. Then $\vxb \in B_{\epsilon}[\vx^*] \subseteq \setM$, which implies
    that $\Ff{\vxb, \vx^*} \geq V + \delta$. Define a function $w: [0, 1]
    \mapsto \real$ as
    \[
        w(\lambda) \defeq \Ff{\vxb, \vx^* + \lambda \vd}
    \]
    Our definition of $\vRyf[^*]{\cdot}$ implies that $\lim_{\lambda \to 0^+}
    w(\lambda) = \Ff{\vxb, \vx^*} \geq V + \delta$. Therefore, there exists
    $\lambda_0 \in (0, 1)$ such that $w(\lambda_0) \geq V + \frac{\delta}{2}$.
    It follows that
    \begin{align}
    \label{eqn:lpg:maximin-equil-nash:Fineq}
    \begin{split}
        \Ff{\vxb, \vx^* + \lambda_0 \vd}
        &\geq V + \frac{\delta}{2} > V = \Ff{\vx^*, \vx^*}
        \geqN{1} \Ff{\vx^* + \lambda_0 \vd, \vx^* + \lambda_0 \vd} \\
        &\geqN{2} (1-\lambda_0) \Ff{\vx^*, \vx^* + \lambda_0 \vd}
        + \lambda_0 \Ff{\vxb, \vx^* + \lambda_0 \vd} \\
        &> (1 - \lambda_0) \Ff{\vx^*, \vx^* + \lambda_0 \vd} + \lambda_0 V
    \end{split}
    \end{align}
    where $\geqN{1}$ follows from property 3 above, and $\geqN{2}$ is due to
    $\Ff{\vx, \vx'}$ being concave in $\vx$. \Cref{
    eqn:lpg:maximin-equil-nash:Fineq} implies that $\Ff{\vx^*, \vx^* +
    \lambda_0 \vd} < V = \Ff{\vx^*, \vx^*}$, which contradicts property 5 of
    $\Ff{\cdot}$.
\end{proof}
\endgroup


\section{Analysis of limited-rank games}
\label{apx:sec:limited-rank-games}

\subsection{Problem setup}
\label{sec:lpg:lrk-problem-setup}
This section defines the limited-rank games and analyzes their properties. The
truncated SVD is used to derive the payoff perception function in two-player
limited-rank games. We discuss how to deal with the non-uniqueness of the SVD.

Let $\vA \in \real^{m \times n}$ be the true payoff matrix. A player with
capability level $c$ perceives a matrix $\ppfunc{\vA}{c} \defeq
\lowrank{\vA}{\minq{c, m, n}}$. Given $\vA \in \real^{m \times n}$ and $r \in
\intInt{1}{\minq{m, n}}$, $\lowrank{\vA}{r}$ denotes the best rank-$r$
approximation in the sense that it minimizes $\norm{\vA - \lowrank{\vA}{r}}$
given some matrix norm operator $\norm{\cdot}$. A norm is \emph{unitarily
invariant} if it satisfies $\norm{\vA} = \norm{\vU \vA \vV}$ for any
appropriately sized orthogonal matrices $\vU$ and $\vV$ in addition to the
standard properties of matrix norms (triangle inequality $\norm{\vX+\vY} \leq
\norm{\vX} + \norm{\vY}$, absolute homogeneity $\norm{c\vX} = \abs{c}
\norm{X}$, and positive definiteness $\norm{\vX}=0 \implies \vX=\V{0}$).
This section assumes a fixed unitarily invariant norm. Typical examples of
unitarily invariant norms include the \emph{Frobenius norm} ($\norm{\vA}_F
\defeq \sqrt{\sum_{i,\,j} A_{i,\,j}^2}$) and the \emph{spectral norm}
($\norm{\vA}_2 \defeq \sup_{\norm{\vx}_2=1} \norm{\vA \vx}_2$).

For any unitarily invariant norm, the best limited-rank approximation is given
by the truncated SVD of $\vA$ \citep{mirsky1960symmetric}. For a matrix $\vA \in
\real^{m \times n}$, its SVD is given by $\vA = \vU \vSig \vV^T$ where $\vU \in
\real^{m \times m}$ and $\vV \in \real^{n \times n}$ are orthogonal matrices and
$\vSig \in \real^{m \times n}$ is a diagonal matrix with non-negative entries
$\sigma_1 \geq \sigma_2 \geq \cdots \geq \sigma_{\minq{m, n}} \geq 0$. We can
rewrite $\vA$ as $\vA = \sum_{i=1}^{\minq{m, n}} \sigma_i \mcol{\vU}{i}
\T{\mcol{\vV}{i}}$. Its rank-$r$ approximation by truncated SVD is
$\lowrank{\vA}{r} = \sum_{i=1}^{r} \sigma_i \mcol{\vU}{i} \T{\mcol{\vV}{i}}$.
The approximation error is $\norm{\vA - \lowrank{\vA}{r}} = \norm{\diag\sqty{
\sigma_{r+1}, \ldots, \sigma_{\minq{m, n}}}}$ \citep{mirsky1960symmetric}.

The SVD of a real-valued matrix always exists and is unique up to unitary
transformations of the subspace corresponding to repeated singular values
\citep[Theorem 4.1]{ trefethen1997numerical}. More precisely, let $i$ and $j$ be
indices of a group of repeated singular values such that $\sigma_i =
\sigma_{i+1} = \cdots = \sigma_j$. Let $\vU' \defeq \mcol{\vU}{i:j}$ and $\vV'
\defeq \mcol{\vV}{i:j}$ be the corresponding singular vectors. Then any
decomposition that replaces $\vU'$ and $\vV'$ with $\vU' \vQ$ and $\vV' \vQ$ for
appropriately sized orthogonal matrix $\vQ$ is also a valid SVD of $\vA$, and
all SVDs of $\vA$ can be obtained in this way. Of note, the singular values are
unique when ordered in non-increasing order.

Let $\lowrankSet{\vA}{r} \subseteq \real^{m \times n \times 3}$ denote the set
of all possible truncated SVDs of $\vA$ that keep the largest $r$ singular
values with corresponding singular vectors. Formally, $\sqty{\vU, \vSig, \vV}
\in \lowrankSet{\vA}{r}$ if and only if $\xTx{\vU} = \vI_m$, $\xTx{\vV} =
\vI_n$, $\vSig = \diag\sqty{\sigma_1, \ldots, \sigma_r, 0, \ldots, 0}$ with
$\sigma_i \geq \sigma_{i+1} \geq 0$ for $i \in \intInt{1}{r-1}$, and there
exists non-negative numbers $\cqty{\sigma_{r+1}, \ldots, \sigma_{\minq{m, n}}}$
such that $\sigma_k \geq \sigma_{k+1} \geq 0$ for $k \in \intInt{r}{\minq{m,
n}}$ and $\vA = \vU \diag\sqty{\sigma_1, \ldots, \sigma_{\minq{m, n}}} \T{\vV}$.

\begin{proposition}
    \label{thm:lowrank-set-compact}
    For any $\vA \in \real^{m \times n}$ and $r \in \intInt{1}{\minq{m, n}}$,
    the set $\lowrankSet{\vA}{r}$ endowed with any norm topology is nonempty
    and compact.
\end{proposition}
\begin{proof}
    The nonemptiness of $\lowrankSet{\vA}{r}$ follows from the existence of the
    SVD of $\vA$. Take an arbitrary item $\sqty{\vU, \vSig, \vV} \in
    \lowrankSet{\vA}{r}$. Then all other items in $\lowrankSet{\vA}{r}$ can be
    obtained by replacing $\vU$ and $\vV$ with $\vU \vQ$ and $\vV \vQ$ for a
    matrix $\vQ$ that is block-orthogonal where the blocks correspond to the
    repeated singular values in $\vSig$. The mapping $\vQ \mapsto \sqty{ \vU\vQ,
    \vSig, \vV\vQ}$ is continuous. Since the set of orthogonal matrices is the
    inverse image of the continuous function $\vM \mapsto \xTx{\vM}$ at $\vI$,
    it is closed. Since the column vectors of an orthogonal matrix are
    orthonormal, the set of orthogonal matrices is bounded. Thus
    $\lowrankSet{\vA}{r}$ is also compact.
\end{proof}

\begingroup
\newcommand{\pcfSym}{\mathfrak{f}}
\newcommand{\pcf}[1]{\pcfSym\sqty{#1}}

There are many ways to define a unique value of $\lowrank{\vA}{r}$ such that the
derived payoff perception function is valid and odd (cf. \cref{ def:lpg:ppf,
def:lpg:odd}). To characterize the possible choices, define the following
\emph{perception-compatible choice function} on $\lowrankSet{\vA}{r}$:
\begin{definition}[Perception-compatible choice function]
    \label{def:lpg:perception-compatible-choice}
    A function $\pcfSym: 2^{\real^{m \times n \times 3}} \to \real^{m \times n}$
    is a \emph{perception-compatible choice function} if and only if it
    satisfies the following properties:
    \begin{itemize}
        \item For $\vA \in \real^{m \times n}$ and $r \in \intInt{1}{\minq{m,
            n}}$, $\pcf{\lowrankSet{\vA}{r}} \in \lowrankSet{\vA}{r}$.
        \item For $\vA \in \real^{m \times n}$ and $1 \leq r_1 < r_2 \leq
            \minq{m, n}$, \[
                \pcf{\lowrankSet{\vA}{r_1}} =
            \pcf{\lowrankSet{\pcf{\lowrankSet{\vA}{r_2}}}{r_1}}
        \]
        \item For $\vA \in \real^{m \times n}$ and $r \in \intInt{1}{\minq{m,
            n}}$, $\pcf{\lowrankSet{-\vA}{r}} = -\pcf{\lowrankSet{\vA}{r}}$.
    \end{itemize}
\end{definition}
It easy to verify that setting $\ppfunc{\vA}{c} = \lowrank{\vA}{ \minq{c, m, n}}
= \pcf{\lowrankSet{\vA}{\minq{c, m, n}}}$ for any perception-compatible choice
function $\pcf{\cdot}$ yields an odd payoff perception function $\ppfunc{\cdot}{
\cdot}$.

One way to define a perception-compatible choice function is to impose a
preference over the basis vectors in the SVD. For example, define
\begin{align*}
    p_i^j(S) &\defeq  \condSet{ \sqty{\vU, \vSig, \vV} \in S }{
        U_{j,\,i} = \inf_{\sqty{\vU', \vSig', \vV'} \in S} U'_{j,\,i}
    }, \\
    q_i(S) &\defeq \qty(p_i^1 \circ \cdots \circ p_i^m)(S),\quad
    f(S) \defeq \qty(q_1 \circ \cdots \circ q_m)(S),
\end{align*}
where $(f \circ g)(x)$ denotes function composition $f(g(x))$. Since $U_{i,\,j}$
is continuous in $\vU$, the set $p_i^j(S)$ is nonempty and compact if $S$ is
nonempty and compact. Intuitively, the function $f(S)$ selects the SVD whose
$\vU$ has the smallest lexicographically ordered columns. One can verify that if
$\sqty{\vU, \vSig, \vV} \in q_i(S)$ and $\sqty{\vU', \vSig', \vV'} \in q_i(S)$,
then $\mcol{\vU}{i} = \mcol{\vU'}{i}$; with similar reasoning, if $\sqty{\vU,
\vSig, \vV} \in f(S)$ and $\sqty{\vU', \vSig', \vV'} \in f(S)$, then $\vU =
\vU'$. Since $\vSig$ is unique and $\vV$ is uniquely determined by $\vU$ via
$\T{\vV} = \vSig^{-1} \T{\vU} \vA$, $f(S)$ contains exactly one item. Define
$f(S) = \pcf{S}$ for any nonempty and compact $S$. Then $\pcf{\cdot}$ is a
perception-compatible choice function.

Given the above results, the limited-rank games are formally defined as follows:
\begin{definition}[Limited-rank game]
    \label{def:lpg:limited-rank-game}
    A limited-rank game is a two-player \oslpGame{} (\cref{ def:lpg:gameplay})
    where the payoff perception function is defined as
    $\ppfunc{\vA}{c} = \pcf{\lowrankSet{\vA}{\minq{c, m, n}}}$ for some
    perception-compatible choice function $\pcf{\cdot}$, and $m \defeq \card{
    \setS_1}$ and $n \defeq \card{\setS_2}$ are the numbers of strategies for
    the row and column players, respectively. The first player is the row player
    and the second player is the column player.

    In an instance of a limited-rank game, the perceived payoff functions are
    typically denoted by matrix pairs $\sqty{\vA_1, \vB_1}$ and $\sqty{\vA_2,
    \vB_2}$ for the row and column players, respectively. If the game is
    zero-sum, then the perceived payoff functions are denoted by $\sqty{\vA_1,
    -\vA_1}$ and $\sqty{\vA_2, -\vA_2}$ for the row and column players,
    respectively.
\end{definition}
\endgroup

There is no known closed-form formula for the SVD of arbitrarily sized matrices.
Practical algorithms use iterative methods to compute the SVD
numerically~\citep{ trefethen1997numerical}.

\subsection{Payoff bounds of limited-rank games}
\label{sec:lpg:lrk-payoff-bound}
This section analyzes the bounds of true payoff given a perceived payoff
function in limited-rank games. The analysis uses the following notation:
\begin{itemize}
    \item For $r = \extPosInt$ and $\vA \in \real^{m \times n}$, $\sigma_r(\vA)$
        denotes the $r$-th largest singular value of $\vA$. Note that the value
        is unique. If $r > \minq{m, n}$ (including $r = \infty$), then
        $\sigma_r(\vA) = 0$.
    \item For $\vA \in \real^{m \times n}$, $\Null{\vA}$ is a matrix whose
        columns form an orthonormal basis for the null space of $\vA$. More
        precisely, $\Null{\vA} = \mqty[\vb_1 & \cdots & \vb_k] \in \real^{n
        \times k}$ where $\vb_i \in \real^n$, $\T{\vb_i} \vb_j =
        \indicator{i=j}$ for $i, j \in \intInt{1}{k}$, and $\vA\vx = 0$ if and
        only if $\vx \in \spn{\vb_1, \ldots, \vb_k}$. Specially, if $\rank(\vA)
        = m$, then $\Null{\vA} = \V{0}$.
\end{itemize}

The main result of this section is the following theorem:

\begin{theorem}[Payoff bounds of limited-rank games]
    \label{thm:lpg:lrk-payoff-bound}
    Given a matrix $\vA \in \real^{m\times n}$ and two vectors $\vx \in \real^m$
    and $\vy \in \real^n$, if $\vA$ is the perceived payoff function of a player
    with capability level $c$ in a limited-rank game (\cref{
    def:lpg:limited-rank-game}), then the true payoff bounds as defined in
    \cref{ def:lpg:payoff-bounds} are given by:
    \begin{align}
    \label{eqn:lpg:lrk-payoff-bound}
    \begin{split}
        \minPayoff{\vA, c, \vx, \vy} &= \T{\vx} \vA \vy
            - \uncPayoff{\vA, c, \vx, \vy} \\
        \maxPayoff{\vA, c, \vx, \vy} &= \T{\vx} \vA \vy
            + \uncPayoff{\vA, c, \vx, \vy} \\
        \text{where } \uncPayoff{\vA, c, \vx, \vy}
            &\defeq \sigma_c(\vA) \cdot
            \norm{\leftNullT{\vA}\vx}_2 \cdot \norm{\NullT{\vA}\vy}_2
    \end{split}
    \end{align}

    Moreover, the limited-rank game is narrowly reversible (\cref{
    def:lpg:narrowly-reversible}).
\end{theorem}

The rest of this section is devoted to proving \cref{ thm:lpg:lrk-payoff-bound}.

\begin{lemma}
    \label{thm:lpg:lrk-approx-error}
    For a matrix $\vA \in \real^{m \times n}$ and an integer $r \in \intInt{1}{
    \minq{m, n} - 1}$, let $\lowrank{\vA}{r} \in \lowrankSet{\vA}{r}$ be a
    rank-$r$ approximation (not necessarily the result of a
    perception-compatible choice function).

    Given any $\vx \in \real^m$ and $\vy \in \real^n$, it holds that
    \begin{align}
    \label{eqn:lpg:lrk-approx-error}
    \begin{split}
        -s &\leq \T{\vx} \vA \vy - \T{\vx} \lowrank{\vA}{r} \vy \leq s \\
        \text{where } s &\defeq \sigma_{r+1}(\vA) \cdot
        \norm{\NullT{\lowrankT{\vA}{r}} \vx}_2 \cdot
        \norm{\NullT{\lowrank{\vA}{r}} \vy}_2
    \end{split}
    \end{align}
    The bound can be relaxed to $s \leq \sigma_{r+1}(\vA) \cdot \norm{\vx}_2
    \cdot \norm {\vy}_2$.
\end{lemma}
\begin{proof}
    Let $f \defeq \T{\vx} \qty(\vA - \lowrank{\vA}{r}) \vy$ be the value to be
    bounded. The following proof only deals with the upper bound of $f$. The
    lower bound can be proved similarly.

    Let $\vA = \vU \vSig \T{\vV}$ be the SVD of $\vA$ that is truncated to
    obtain $\lowrank{\vA}{r}$. Then we have $f = \sum_{i=r+1}^{\minq{m,n}}
    \sigma_i(\vA) \T{\vx} \mcol{\vU}{i} \T{\mcol{\vV}{i}} \vy$. The value of $f$
    can be increased if we negate some columns of $\vU$ to ensure that $\T{\vx}
    \vU_{:,\,i} \T{\vV_{:, \,i}} \vy > 0$ for $r + 1 \leq i \leq \minq{m, n}$
    and set $\sigma_i(\vA) = \sigma_{r+1}(\vA)$ for $r + 2 \leq i \leq \minq{m,
    n}$. More precisely, we have $f \leq f^*$ where
    \begin{align*}
        f^* &\defeq \T{\vx} \qty(\vA^* - \lowrank{\vA}{r}) \vy,\quad
        \vA^* \defeq \vU\vS \vSig^* \T{\vV},\quad
        \vS \defeq \diag\sqty{s_1, s_2, \ldots, s_{\minq{m, n}}} \\
        s_i &\defeq \begin{dcases}
            -1 & \text{if } i \geq r + 1
                \text{ and } \T{\vx}\mcol{\vU}{i}\T{\mcol{\vV}{i}}\vy < 0 \\
            1 & \text{otherwise}
        \end{dcases} \\
        \vSig^* &\defeq \diag\sqty{\sigma_1(\vA), \sigma_2(\vA), \ldots,
            \sigma_{r+1}(\vA), \sigma_{r+1}(\vA), \ldots, \sigma_{r+1}(\vA)}
    \end{align*}

    Let $\vX_k \defeq \mcol{\vU}{k:m}$ be the $k$-th and subsequent columns of
    $\vU$. Similarly define $\vY_k \defeq \mcol{\vV}{k:n}$ and $\vX_k^* \defeq
    \vU \mcol{\vS}{k:m}$. Since $\vS$ is a diagonal matrix, we have $\vX_k^* =
    \vX_k \vS_{k:m,\,k:m}$. Rewrite $f^*$ as:
    \begin{align}
    \label{eqn:lpg:lrk-approx-error-fstar}
    \begin{split}
        f^* &= \T{\vx}\vU\vS \qty( \vSig^* - \diag\sqty{\sigma_1(\vA),
            \sigma_2(\vA), \ldots, \sigma_r(\vA)}) \T{\vV} \vy \\
            &= \T{\vx}\vX_{r+1}^* \sigma_{r+1}(\vA)\vI^{(m-r)\times(n-r)}
            \T{\vY_{r+1}} \vy \\
            &= \sigma_{r+1}(\vA) \T{\qty(\T[*]{\vX_{r+1}}\vx)}
                \qty(\vI^{(m-r)\times(n-r)} \T{\vY_{r+1}} \vy ) \\
            &\leqN1 \sigma_{r+1}(\vA) \norm{\T[*]{\vX_{r+1}}\vx}_2
                \norm{\vI^{(m-r)\times(n-r)} \T{\vY_{r+1}} \vy }_2 \\
            &\leq \sigma_{r+1}(\vA) \norm{\T[*]{\vX_{r+1}}\vx}_2
                \norm{\T{\vY_{r+1}} \vy }_2 \\
            &= \sigma_{r+1}(\vA)
                \norm{\T{\vS_{r+1:m,\,r+1:m}} \T{\vX_{r+1}}\vx}_2
                \norm{\T{\vY_{r+1}} \vy }_2 \\
            &= \sigma_{r+1}(\vA) \norm{\T{\vX_{r+1}}\vx}_2
                \norm{\T{\vY_{r+1}} \vy }_2, \\
    \end{split}
    \end{align}
    where $\leqN1$ follows from the Cauchy-Schwarz inequality. Since $\vU$ and
    $\vV$ are orthogonal matrices, \cref{ eqn:lpg:lrk-approx-error} can
    be proven by the plugging following inequalities into \cref{
    eqn:lpg:lrk-approx-error-fstar}:
    \begin{align*}
    \begin{gathered}
        \vX_{r+1} = \leftNull{\vU_{:,\,1:r}}
            = \Null{\lowrankT{\vA}{r}},\quad
        \norm{\T{\vX_{r+1}}\vx}_2 \leq \norm{\vx}_2 \\
        \vY_{r+1} = \leftNull{\vV_{:,\,1:r}}
            = \Null{\lowrank{\vA}{r}},\quad
        \norm{\T{\vY_{r+1}}\vy}_2 \leq \norm{\vy}_2
    \end{gathered}
    \end{align*}
\end{proof}

\begingroup
\newcommand{\vAm}[1][\epsilon]{\vA^-_{#1}}
\newcommand{\vAp}[1][\epsilon]{\vA^+_{#1}}
\begin{lemma}
    \label{thm:lpg:lrk-payoff-bound-lem}
    Given a rank-$r$ matrix $\vA \in \real^{m\times n}$ and two vectors $\vx \in
    \real^m$ and $\vy \in \real^n$, the true payoff bounds satisfy:
    \begin{align}
    \label{eqn:lpg:lrk-payoff-bound-lem}
    \begin{split}
        \minPayoff{\vA, r, \vx, \vy} &= \T{\vx} \vA \vy
            - \uncPayoff{\vA, r, \vx, \vy} \\
        \maxPayoff{\vA, r, \vx, \vy} &= \T{\vx} \vA \vy
            + \uncPayoff{\vA, r, \vx, \vy} \\
    \end{split}
    \end{align}

    When $r < \minq{m, n}$, for any $\epsilon \in \posReal$, there exist
    matrices $\vAm$ and $\vAp$ such that
    \begin{align}
    \label{eqn:lpg:lrk-payoff-bound-tight}
    \begin{gathered}
        \rank(\vAm) = \rank(\vAp) = r + 1,\quad
        \lowrankSet{\vAm}{r} = \lowrankSet{\vAp}{r} = \cqty{\vA} \\
        \minPayoff{\vAm, r+1, \vx, \vy} = \maxPayoff{\vAm, r+1, \vx, \vy}
            = \T{\vx} \vA \vy - (1 - \epsilon)s \\
        \minPayoff{\vAp, r+1, \vx, \vy} = \maxPayoff{\vAp, r+1, \vx, \vy}
            = \T{\vx} \vA \vy + (1 - \epsilon)s \\
        \text{where } s \defeq \uncPayoff{\vA, r, \vx, \vy}
    \end{gathered}
    \end{align}
\end{lemma}
\begin{proof}
    \newcommand{\vAh}{\vA'}
    \newcommand{\vUr}{\widetilde{\vU}}
    \newcommand{\vVr}{\widetilde{\vV}}
    Let $s \defeq \uncPayoff{\vA, r, \vx, \vy}$. When $\vA$ is full-rank (i.e.,
    $r = \minq{m, n}$), either $\leftNull{\vA}$ or $\Null{\vA}$ is the zero
    vector, which implies $s = 0$. Therefore, \cref{
    eqn:lpg:lrk-payoff-bound-lem} holds when $r = \minq{m, n}$.

    Now assume $r < \minq{m, n}$. Take an arbitrary $\vAh \in
    \ppfuncInv{\vA}{r}$. Since $\sigma_{r+1}(\vAh) \leq \sigma_r(\vAh) =
    \sigma_r(\vA)$, \cref{thm:lpg:lrk-approx-error} yields $\abs{\T{\vx} (\vA -
    \vAh) \vy} \leq s$, which implies
    \begin{align}
        \label{eqn:lpg:lrk-payoff-bound:proof-1}
        \minPayoff{\vA, r, \vx, \vy} \geq \T{\vx} \vA \vy - s,\quad
        \maxPayoff{\vA, r, \vx, \vy} \leq \T{\vx} \vA \vy + s
    \end{align}
    \Cref{eqn:lpg:lrk-payoff-bound:proof-1,
    eqn:lpg:lrk-payoff-bound-tight} together imply \cref{
    eqn:lpg:lrk-payoff-bound-lem}. The rest of the proof focuses on
    \cref{eqn:lpg:lrk-payoff-bound-tight}.

    Assume the SVD of $\vA$ is $\vA = \vU \vSig \T{\vV}$ where $\Sigma_{i,\,i} =
    0$ for $i \geq r+1$. Define $\vUr \defeq \mcol{\vU}{1:r}$ and $\vVr \defeq
    \mcol{\vV}{1:r}$ to be the left or right singular vectors of $\vA$
    corresponding to nonzero singular values. Decompose $\vx$ as $\vx = \vUr
    \vx_0 + \vx_b$ where $\vx_0 \defeq \T{\vUr}\vx$ so that $\vx_b$ is
    orthogonal to the column space of $\vUr$ (i.e., $\T{\vUr}\vx_b = \V{0}$). It
    follows that $\vx_b$ is in the column space of $\mcol{\vU}{r+1:m}$, which
    implies that there is an orthogonal matrix $\vP$ such that
    $\T{\vx_b}\qty(\mcol{\vU}{r+1:m}\vP) = \mqty[\norm{\vx_b}_2 & \V{0}]$. It
    also holds that $\norm{\vx_b}_2 = \norm{\leftNullT{\vA}{\vx}}_2$. Similarly,
    decompose $\vy$ as $\vy = \vVr \vy_0 + \vy_b$ and find an orthogonal matrix
    $\vQ$ such that $\T{\qty( \mcol{\vV}{r+1:n}\vQ)}\vy_b =
    \T{\mqty[\norm{\vy_b}_2 & \V{0}]}$.

    For $q \in \real$, define
    \begin{align*}
        \vA_q \defeq \mqty[\vUr & \mcol{\vU}{r+1:m}\vP]
        \diag\sqty{\sigma_1\qty( \vA), \ldots, \sigma_r\qty(\vA), q,
            0, \ldots, 0} \T{\mqty[\vVr & \mcol{\vV}{r+1:n}\vQ]}
    \end{align*}
    One can verify that for $\abs{q} < \sigma_r\qty(\vA)$, the following holds:
    \begin{align*}
    \begin{gathered}
        \rank(\vA_q) = r + 1,\quad
        \lowrankSet{\vA_q}{r} = \cqty{\vA},\quad
        \uncPayoff{\vA_q, r + 1, \vx, \vy} = 0 \\
        \T{\vx} \vA_q \vy = \T{\vx} \vA \vy + q \norm{\vx_b}_2 \norm{\vy_b}_2
        = \T{\vx} \vA \vy + q \cdot
            \norm{\leftNullT{\vA}\vx}_2 \cdot \norm{\NullT{\vA}\vy}_2
    \end{gathered}
    \end{align*}
    Therefore, \cref{ eqn:lpg:lrk-payoff-bound-tight} is satisfied by
    setting $\vAm = \vA_{-w}$ and $\vAp = \vA_{w}$ where $w \defeq (1 -
    \epsilon) \sigma_r\qty(\vA)$.
\end{proof}

\begin{separateProof}{thm:lpg:lrk-payoff-bound}
    Since $\vA$ is the perceived payoff function of a player with capability
    level $c$ in a limited-rank game, we have $\rank(\vA) \leq c$. If
    $\rank(\vA) = c$, then \cref{ eqn:lpg:lrk-payoff-bound-lem} directly implies
    \cref{ eqn:lpg:lrk-payoff-bound}. Otherwise when $\rank(\vA) < c$, we have
    $\sigma_c(\vA) = 0$ and $\ppfuncInv{\vA}{c} = \cqty{\vA}$, which gives
    $\uncPayoff{\vA, c, \vx, \vy} = 0$ and thus \cref{ eqn:lpg:lrk-payoff-bound}
    still holds.

    Let's proceed to show that \cref{ eqn:lpg:lower-narrowly-reversible } holds
    for $c_1 = c$ and $c_2 = c + 1$, and therefore the limited-rank game is
    narrowly reversible. Take arbitrary values of $\vx \in \simplex{m}$ and $\vy
    \in \simplex{n}$. If $ \rank(\vA) = c$, then take an arbitrary sequence
    $\epsilon_i \to 0$. \Cref{ thm:lpg:lrk-payoff-bound-lem} implies that there
    exists $\vAm[\epsilon_i] \in \ncSet{\vA, c, c+1}$ satisfying \cref{
    eqn:lpg:lrk-payoff-bound-tight}, which yields
    \begin{align*}
        \inf_{\vA' \in \ncSet{\vA, c, c+1}} \maxPayoff{\vA', c+1, \vx, \vy}
        &\leq \inf_{i \to \infty} \maxPayoff{\vAm[\epsilon_i], c+1, \vx, \vy} \\
        &= \inf_{i \to \infty} \qty( \T{\vx}\vA\vy - (1 - \epsilon_i) s)
        = \minPayoff{\vA, c, \vx, \vy}
    \end{align*}
    One can also easily verify that \[
        \inf_{\vA' \in \ncSet{\vA, c, c+1}} \maxPayoff{\vA', c+1, \vx, \vy}
        \geq \minPayoff{\vA, c, \vx, \vy}
    \]
    Therefore,
    \begin{align*}
        \inf_{\vA' \in \ncSet{\vA, c, c+1}} \maxPayoff{\vA', c+1, \vx, \vy}
        = \minPayoff{\vA, c, \vx, \vy}
    \end{align*}
    Analogously, one can prove $\sup_{\vA' \in \ncSet{\vA, c, c+1}}
    \minPayoff{\vA', c+1, \vx, \vy} = \maxPayoff{\vA, c, \vx, \vy}$.

    If $\rank(\vA) < c$, then
    \begin{align*}
        \sup_{\vA' \in \ncSet{\vA, c, c+1}} \minPayoff{\vA', c+1, \vx, \vy}
        &= \inf_{\vA' \in \ncSet{\vA, c, c+1}}
            \maxPayoff{\vA', c+1, \vx, \vy} \\
        &= \minPayoff{\vA, c, \vx, \vy}
        = \maxPayoff{\vA, c, \vx, \vy}
        = \T{\vx}\vA\vy
    \end{align*}

    Therefore, the limited-rank game is narrowly reversible.
\end{separateProof}
\endgroup

\end{document}